%% file: main.tex
\documentclass{article}
\usepackage[margin=1in]{geometry}
\usepackage[utf8]{inputenc}
\usepackage[T1]{fontenc}
\usepackage{amsmath, amssymb}
\usepackage{amsthm, thm-restate}
\usepackage{hyperref, cleveref}
\usepackage{multirow}

\newtheorem{theorem}{Theorem}
\newtheorem{claim}[theorem]{Claim}
\newtheorem{lemma}[theorem]{Lemma}
\newtheorem{observation}[theorem]{Observation}
\newtheorem{definition}[theorem]{Definition}
\newtheorem{remark}[theorem]{Remark}
\newtheorem{corollary}[theorem]{Corollary}

\newcommand\calC{\mathcal{C}}
\newcommand\calE{\mathcal{E}}

\newcommand\calM{\mathcal{M}}
\newcommand\calP{\mathcal{P}}
\newcommand\calR{\mathcal{R}}
\newcommand\calS{\mathcal{S}}
\newcommand\calX{\mathcal{X}}
\newcommand\calT{\mathcal{T}}
\newcommand\eps\varepsilon

\newcommand\bc{\boldsymbol c}
\newcommand\bi{\boldsymbol i}
\newcommand\bo{\boldsymbol o}

\newcommand\half{\frac 1 2}
\newcommand\EE{\mathbb E}
\newcommand\NN{\mathbb N}
\newcommand\RR{\mathbb{R}}
\newcommand{\ZZ}{\mathbb{Z}}
\def\cT{\calT}
\def\e{\eps}
\newcommand\restrict[1]{\raisebox{-.5ex}{$|$}_{#1}}
\newcommand\abs[1]{\left\lvert#1\right\rvert}
\newcommand\bangle[1]{\left\langle#1\right\rangle}

\newcommand\paren[1]{\left(#1\right)}

\newcommand\set[1]{\left\{#1\right\}}
\newcommand\ALG{\mathtt{ALG}}
\newcommand\child{\mathsf{child}}

\newcommand\DP{\mathrm{DP}}
\newcommand\cost{\mathrm{cost}}
\newcommand\aggcost{\mathrm{aggcost}}

\newcommand\lp{{\sf lp}}
\newcommand\LP{{\sf LP}}
\newcommand\nbr{\mathsf{nbr}}
\DeclareMathOperator*{\EXP}{\EE}
\newcommand\opt{\mathtt{opt}}
\newcommand\OPT{\mathtt{OPT}}

\newcommand\poly{\mathrm{poly}}
\newcommand\prim{\mathsf{prim}}

\renewcommand\sec{\mathsf{sec}}
\newcommand\Reps{\mathsf{Reps}}
\newcommand\Super{\mathsf{Super}}
\newcommand\tw{\mathsf{tw}}
\def\etal{\emph{et~al.}}			
\def\dist{d}

\newtheorem{assumption}[theorem]{Assumption}
\newtheorem{fact}[theorem]{Fact}

\usepackage{mathtools}
\usepackage{xcolor}  
\definecolor{TODOcolor}{cmyk}{0.05,0,0,0}
\RequirePackage[color=red, backgroundcolor=TODOcolor, textcolor=red, linecolor=lightgray, bordercolor=red, textsize=tiny, tickmarkheight=3pt, figwidth=0.5\linewidth, figheight=4cm, figcolor=TODOcolor]{todonotes}

\title{Clustering in Varying Metrics}

\author{
Deeparnab Chakrabarty\thanks{Department of Computer Science, Dartmouth College. \\ \texttt{\{deeparnab, jonathan.conroy.gr, ankita.sarkar.gr\}@dartmouth.edu}}
\and
Jonathan Conroy\footnotemark[1]
\and
Ankita Sarkar\footnotemark[1]
}

\date{}

\begin{document}

\maketitle
\begin{abstract}
We introduce the {\em aggregated clustering} problem, where one is given $T$ instances of a center-based clustering task over the same $n$ points, but under different metrics. The goal is to open $k$ centers to minimize an aggregate of the clustering costs—e.g., the average or maximum—where the cost is measured via $k$-center/median/means objectives. More generally, we minimize a norm $\Psi$ over the $T$ cost values.

We show that for $T \geq 3$, the problem is inapproximable to any finite factor in polynomial time. For $T = 2$, we give constant-factor approximations. We also show W[2]-hardness when parameterized by $k$, but obtain $f(k,T)\poly(n)$-time 3-approximations when parameterized by both $k$ and $T$.

When the metrics have structure, we obtain efficient parameterized approximation schemes (EPAS). If all $T$ metrics have bounded $\varepsilon$-scatter dimension, we achieve a $(1+\varepsilon)$-approximation in $f(k,T,\varepsilon)\poly(n)$ time. If the metrics are induced by edge weights on a common graph $G$ of bounded treewidth $\tw$, and $\Psi$ is the sum function, we get an EPAS in $f(T,\varepsilon,\tw)\poly(n,k)$ time. Conversely, unless (randomized) ETH is false, any finite factor approximation is impossible if parametrized by only $T$, even when the treewidth is $\tw = \Omega(\poly\log n)$.
\end{abstract}

\section{Introduction}

Clustering problems such as the $k$-supplier or the $k$-median problem are classic discrete optimization problems which have formed the bedrock of many problems in operations research.
In its simplest form, the objective is to ``open'' $k$ locations in a metric space such that some function (eg, max or sum) of the distances of clients to the nearest open facility is minimized.
Such problems have been extensively studied over the past 50 years \cite{HochbS1986,Ples1987,ChariGMN2001,JainV2001,ChariGTS2002,AryaGKMMP2004,ChariL2012,ByrkaPRST2014,LiS2016,Swamy2016,ChakrN2019,KrishLS2018,ChakrGK2020,MahabV2020,MakarV2021,ChakrNS2022,ChlamMV2022,DengLR2022,AgrawISX2023,ChakrCS2024}, and we have a very good understanding on the approximability of many of these problems.

In this paper, we consider clustering problems \emph{when the distance function underlying the metric space can change over time}. To give a toy example, imagine that the clients are households in a 
New England town, and the distances are travel times over the road network in the town. On one hand, travel times may be quite different in January and June given the weather differences over the months.
On the other hand, one can't hope to open different locations in different seasons. Our decision ought to take into account the different costs over time, for example by selecting a set of facilities that minimizes the average cost over the year. Although this sounds a natural question, 
to our knowledge what happens to the complexity of clustering with changing metrics hasn't been studied so far.\footnote{One may be tempted to simply perform clustering on a single metric space defined by average distances over all months. However, this does not work: the cost of a client at each timestep is computed by take a minimum over a set of facilities, and this minimum is not preserved when we take an average.}
Formally, we study the following problem.

\begin{definition}[Aggregate Clustering Problems]
	We are given $V = F\cup C$, a finite set of $n$ points where $C$ is the set of clients and $F$ is a set of facilities\footnote{$F$ and $C$ need not be disjoint, and one can imagine the case of $F =C$ as a special but instructive case.}. We are given $T$ distance functions $d_t(\cdot, \cdot)$ over $V\times V$ which are symmetric and satisfy triangle-inequality. 
	We are given a parameter $k \in \ZZ$. Finally, we are given an {\em aggregator} function $\Psi: \RR^T_{\geq 0} \to \RR_{\geq 0}$. The goal is to `open' a subset $S\subseteq F$ of $k$ facilities that minimize a certain objective function (depending on the problem variant).
	The objective functions are:

	\begin{itemize}
		\item $\Psi\left(\cost_\infty(d_t;S):t\in [T]\right)$, where $\cost_\infty(d_t;S) := \max_{j\in C} d_t(j, S)$, in the $\Psi$-aggregate $k$-supplier problem.\footnote{We use the notation $\dist_t(j,S)$ to mean $\min_{i \in S}\dist_{t}(j,i)$.}
		\item $\Psi\left(\cost_1(d_t;S):t\in [T]\right)$, where $\cost_1(d_t;S) := \sum_{j\in C} d_t(j, S)$, in the $\Psi$-aggregate $k$-median problem.
		\item $\Psi\left(\cost_z(d_t;S):t\in [T]\right)$, where $\cost_z(d_t;S) := \left(\sum_{j\in C} d_t(j, S)^z\right)^{1/z}$, in the $\Psi$-aggregate $(k,z)$-clustering problem.
	\end{itemize}
    When it is clear from context, we denote the objective function as $\aggcost(\cdot)$.
	
\end{definition}

The above problem also can be thought of as a stochastic/robust version of clustering under uncertainty where the uncertainty is over the distances. Note that when $\Psi$ is the sum-function, 
we wish to find a solution which minimizes the average $k$-supplier/$k$-median cost of the returned solution $S$, where the average is over the $T$ ``scenarios'' (we are assuming the full-information setting); when $\Psi$ is the max-function, we get the robust optimization setting. In the literature~\cite{AnthoGGN2010, MakarV2021, ChlamMV2022}, stochastic/robust clustering has been studied when the metric is fixed but the uncertainty is over the client sets. 
We can consider a generalized version of aggregate clustering, which models client demands changing over time.
\begin{definition}[Generalized Aggregate Clustering]
	Apart from the $t$ distances, we have $t$ weight functions $w_t: C \to \RR$. In the objective functions, we modify 
	$\cost$s by multiplying $d_t(j,S)$'s with $w_t(j)$. More precisely, $\cost_z(d_t,S) = \left(\sum_{j\in C} w_t(j) \cdot d_t(j, S)^z\right)^{1/z}$ 
	and $\cost_\infty(d_t;S) = \max_{j\in C} w_t(j) d_t(j,S)$. In particular, if $w_t: C\to \{0,1\}$, then this models having different client sets at different times.
\end{definition}
All our upper bounds apply for generalized aggregate clustering, and all our lower bounds hold even for the vanilla version of aggregate clustering. In fact, in~\Cref{obs:generalized} we present a weak equivalence between the two problems, by showing that any algorithm for vanilla aggregate clustering can solve generalized aggregate clustering with $\set{0,1}$ weights, i.e. aggregate clustering where the client set is allowed to change over time.

\subsection{Our Findings}
\textbf{General Metrics.}
In its full generality, unfortunately, the aggregrate clustering problem becomes very hard when the number of scenarios $T$ becomes $3$ or larger. In~\Cref{thm:T=3-hardness} we show that it is NP-hard to distinguish between
zero and non-zero solutions when $T \geq 3$. To complement this result, we show (in \Cref{thm:k-supp:3appx,thm:k-med:const-appx}) that when $T = 2$, both the (generalized) $\Psi$-aggregate $k$-median and $k$-supplier problems have polynomial time $O(1)$-approximations when $\Psi$ is any norm. 
In fact, for $k$-supplier, we match the optimal $3$-approximation that is known at $T=1$~\cite{HochbS1986,Ples1987}. Both the above results, at a very high level, follow from the fact that bipartite matching (or rather, partition matroid intersection) is tractable, while 3D-matching is NP-hard.

Given the above hardness for $T \geq 3$, we attempt to bypass it in two ways: (i) investigating fixed-parameter tractable approximation algorithms, and (ii) imposing additional structure over the metrics in various scenarios. 

\medskip \noindent \textbf{Constant-factor FPT approximation algorithms.} The hardness for $T=3$ clearly shows one can't expect $f(T)\cdot\poly(n,k)$-time approximation algorithms. Moreover, we show (in \Cref{thm:w2-hardness}) that it is also hard to obtain $f(k)\cdot\poly(n,T)$-time approximation algorithms: we encode the $k$-hitting set problem as an instance of aggregate clustering, and our lower bound follows from the $W[2]$ hardness of $k$-hitting set~\cite[Theorem 13.12]{CyganFKLMPPS2015}. On the other hand, if we allow $k$ and $T$ to be both fixed parameters, then (in \Cref{thm:3-appx}) we obtain an $f(k,T)\cdot \poly(n)$-time $(3+\e)$-approximation. The high level idea for this is as follows: Suppose we have $T$ partitions of a universe $U$, and we must decide if there is a $k$-sized subset that is a \emph{simultaneous} hitting set for all $T$ partitions. In $(k+1)^{kT}\cdot \poly(n)$-time, we can guess for each $t \in [T]$ and $i \in [k]$ the part in the $t$\textsuperscript{th} partition that the $i$\textsuperscript{th} hitting element hits.
Such an idea is the core of many FPT approximation algorithms for clustering~\cite{ABB+23,BGI25,GI25}. 

\medskip \noindent \textbf{Approximation schemes in structured metrics.} To impose structure on the metric, we note that any finite metric $d_t$ is a shortest path metric on an undirected {\em base graph} $G = (V, E)$ with weights $w_t(e)$ on the edges.
Note that our restriction of having the {\em same} graph $G$ across scenarios is without loss of generality if we allow a complete graph. 
Furthermore, in many applications, such as the road-network application mentioned in the early paragraphs of the introduction, this indeed is the case where the ``weights'' change over
the seasons but the underlying network remains the same. Can one obtain better algorithmic results exploiting the structure of the underlying undirected base graph $G$?

Unfortunately, our hardness results apply even for very simple base graphs $G$. In particular, our reduction from hitting set (\Cref{thm:w2-hardness}, which rules out any finite factor approximation algorithms in $f(k) \cdot \poly(n,T)$ time) can be carried out even when $G$ is a star. Moreover, our reduction from 3D matching (\Cref{thm:T=3-hardness}, which rules out any $f(T)\cdot\poly(n,k)$-time finite factor approximations) can be carried out even when $G$ is a grid graph (see \Cref{thm:treewidth-lb}).

Nevertheless, we bypass this hardness by designing $1+\eps$ approximation algorithms in time $f(\eps, k, T) \cdot \poly(n)$  (ie, EPASes in parameters $k$ and $T$) for a broad class of structured input metrics, which includes (for example) the case when $G$ is a planar graph. More precisely, we design approximation algorithms whenever input metric $d_t$ has {\em bounded scatter dimension}. This concept was recently introduced by Abbasi \etal~\cite{ABB+23}, along with an algorithm demonstrating that metrics of bounded scatter dimension admit $f(\eps, k)\cdot \poly(n)$ time $1+\eps$-approximation algorithms for many clustering problems. We generalize their ideas to obtain an EPAS for the (generalized) aggregate $(k,z)$-clustering problem. Bounded scatter dimension captures a rich class of metrics, including doubling metrics, shortest-path metrics on planar or minor-free graphs \cite{BP25}, and metrics of {\em bounded highway dimension}\footnote{We are not aware of a published proof that highway dimension has bounded scatter dimension, but the claim can be found at timestamp 0:52 in the recorded FOCS 2023 talk on \cite{ABB+23} (see link https://focs.computer.org/2023/schedule/) and also in \cite{BGI25}.} --- this last class, which was introduced to model transportation networks \cite{abraham2010highway, abraham2016highway}, may be of particular interest in light of our road-network application from earlier. We emphasize that our lower bounds on the star and grid apply to metrics of bounded scatter dimension, so our FPT$(k, T)$ approximation is essentially the best one could hope to achieve (up to the $1+\eps$ approximation) in this setting.

Additionally, we consider the case when $G$ is a tree, or more generally a bounded-treewidth graph. As bounded-treewidth graphs have bounded scatter dimension, our $f(\eps, k, T)\cdot \poly(n)$-time algorithm from above applies. However, we can do even better --- we obtain a $(1+\eps)$-approximation for the sum-aggregate $(k,z)$-clustering problem in time $f(T, \eps) \cdot \poly(n,k)$, by generalizing a folklore dynamic programming
algorithm for $k$-median on bounded-treewidth graphs \cite{ARR98, KR07, CFS21}.
Parameterizing by $T$ is necessary, because of our hardness result on the star. On the other hand, we can avoid parameterizing by $k$ because our hardness for the $T= 3$ case only applies $G$ is a grid graph, which has high treewidth --- recall that graphs of bounded treewidth are precisely those graphs which exclude a grid as a minor \cite{robertson1986graph, chuzhoy2021towards}. In fact, we show in \Cref{thm:treewidth-lb} that our $T=3$ hardness rules out $f(T) \cdot \poly(n)$-time algorithms for \emph{any} graph $G$ that has sufficiently large treewidth (specifically, $\Omega(\log^{295} n)$), whereas our upper bound in \Cref{thm:treewidth} achieves $f(T) \cdot \poly(n)$ runtime so long as $G$ has treewidth $o(\log n/\log \log n)$ (see \Cref{remark:treewidth-algo-limit}). Thus, our upper and lower bounds show that the treewidth of the base graph $G$ almost entirely captures the tractability of FPT($T$) algorithms for the sum-aggregate\footnote{We expect that a DP could be designed for other $\Psi$-aggregator norms, but in this paper we only consider $\Psi=$ sum.}
clustering problem (up to some $1+\eps$ approximation factor, and up to some $\poly(\log n)$ gap in the treewidth between our upper and lower bounds).

\paragraph*{Open Directions}
We end this section mentioning some open directions for future study. If one observes the metrics we use to prove our hardness results, we notice that 
they ``change a lot'' as one moves from $d_t$ to $d_{t+1}$. Can one obtain positive algorithmic results if this change is bounded? Concretely, assume 
$d_t$ is the shortest path metric on $G_t = (V, E_t)$ and $d_{t+1}$ is the same on $G_{t+1} = (V, E_{t+1})$. What if we forced $|E_t \Delta E_{t+1}| \leq O(1)$?
Again taking the road-network example, these could lead to certain road-closures or road-building, but in a single period perhaps not too many edges are added or deleted.
We do not know any lower bound for this version, nor do we know how to design algorithms. Structurally, this seems to require an understanding of how shortest paths change with insertion/deletion of edges
which may be of independent interest.

Another interesting, but slightly less concrete, direction is to explore {\em online} algorithms where an algorithm is allowed to change the centers but the distance metric is revealed only after the algorithms decision is made. Can one design ``low-regret'' algorithms as in~\cite{Zinke03}?
The sum-aggregate objective is in fact the benchmark with which regret is measured, and our results show that approximating this benchmark for general metrics is hard. With this in view, what would be a good definition of regret?

Coming back to aggregate clustering, a concrete open problem is the following. The aforementioned hardness when $G$ is a star actually has $F\neq C$, but the case of $V = C = F$ is also interesting.
Are there $f(T)\cdot\poly(n,k)$ approximation or even exact algorithms for $\Psi$-aggregated $k$-supplier/median problem when $G$ is a tree? We believe this is an interesting generalization of $k$-median problem on trees and worthy of exploration.

\subsection{Related Works}
As noted earlier, approximation algorithms for clustering problems have a rich history and we don't attempt to summarize this. Rather we mention the works most related to the ``changing metric'' viewpoint.

Deng, Li, and Rabani~\cite{DengLR2022} consider the $k$-clustering problems in a dynamic setting in which both clients and facilities may change over time (though the metric remains the same). The algorithm is allowed to change the open facilities over time, and the objective is to minimize the clustering costs and the distances that open facilities move over time. Their main results, at a high level, are in lockstep with our first set of results for aggregate clustering (on the tractability of $T=2$ vs $T \ge 3$ scenarios); though our findings and theirs do not imply one another. Their setting has an ambient metric space $(X,d)$ and, over $T$ timesteps, there are client and facility sets $C_t,F_t \subseteq X$ for each timestep $t$. The goal is to pick $k$-sized solutions $S_t \subseteq F_t$ at each time step minimizing a certain objective function.

As mentioned earlier, our problem is related to stachastic/robust clustering problems which consider
scenarios where client sets change but the algorithm needs to open the same set of $k$ centers to serve these. This question was studied by
Anthony, Goyal, Gupta, and Nagarajan~\cite{AnthoGGN2010} who obtain a $O(\log T + \log n)$-approximation for the robust $k$-median problem discussed earlier. They also show that stochastic $k$-center (i.e. stochastic $k$-supplier with $F = C$) is as hard to approximate as the densest-$k$-subgraph problem. 
Makarychev and Vakilian~\cite{MakarV2021} obtain an $\paren{e^{O(z)}\frac{\log T}{\log \log T}}$-approximation to robust $(k,z)$-clustering. Along with Chlamt{\'a}{\v{c}}, the same authors~\cite{ChlamMV2022} subsequently generalize this approximation algorithm to clustering with cascaded norm objectives, i.e. objectives of the form $\Big(\sum_{t=1}^T\big(\sum_{j \in C_t}d(j,S)^z\big)^{q/z}\Big)^{1/q}$. When $z \geq q$, they obtain a near-optimal $O\paren{k^{O_{z,q}(1)}}$-approximation. When $z \leq q < \infty$, they obtain an $O_{z,q}(1)$-approximation.

Another related line of work is {\em dynamic} algorithms for clustering where clients come and go, and the algorithm needs to update the solution fast and maintain $O(1)$-approximations. There has been a lot of work \cite{LattaV17,ChanGS18,CohenHPSS19,FichtLNS21,BatenEFHJMW23,BhattCLP2023,LackiHGJV24,ForstS25,BhattCLP25,BhattCF25} recently on these probelms. We point the readers to~\cite{BhattCLP25} (and references within) for dynamic $k$-center, and~\cite{BhattCF25}  (and references within) for dynamic $k$-median.
It is curious to see if techniques from these can address the questions mentioned in the open-directions paragraph earlier.

\section{Hardness of Approximation}\label{sec:hardness}

\begin{theorem}[Hardness of approximation when $T\geq 3$]\label{thm:T=3-hardness}
	For any homogeneous aggregator $\Psi$ and any $z \in \NN \cup \set{\infty}$, 
	there is no $f(T)\poly(n,k)$-time finite-factor approximation for $\Psi$-aggregate $(k,z)$-clustering on finite metrics on $n$-vertices unless $P=NP$.
	The result holds even for $T=3$.
\end{theorem}
\begin{proof}
	We show that it is NP-hard to decide whether the optimum value is finite or infinite for any $\Psi$-aggregate $(k,z)$-clustering problem.
    We reduce from the (perfect) 3D Matching problem where, recall, we are given a $3$-partite hypergraph $H = (V_1 \uplus V_2 \uplus V_3, E)$
    where $|V_1| = |V_2| = |V_3|$ any every hyperedge $e\in E$ intersects every $V_t$ exactly once. The goal is to decide whether or not there exists a {\em matching} (pairwise disjoint subset) $M\subseteq E$ 
    of hyperedges which spans all vertices. Equivalently, $|M\cap e| = 1$ for all $e\in E$.

    Given such a 3D Matching instance $(V_1\uplus V_2 \uplus V_3,E)$, we construct an aggregate clustering instance on $T = 3$ metrics, where the underlying graph is a complete graph $G$ with vertex set $E$. The client set $C = F = E$ and $k \coloneqq \abs{V_1}=\abs{V_2}=\abs{V_3}$. For each metric $t \in \set{1,2,3}$ and each two distinct clients $e,e'$, we set the distance $d_t(e,e') \coloneqq 0$ if $e \cap e' \cap V_t \neq \emptyset$, i.e. if they agree on their vertex in $V_t$. Otherwise set $d_t(e,e') \coloneqq \infty$. It is easy to see that each $d_t$ is in fact a metric since if $e_1, e_2$ share the same $V_t$ endpoint as do $e_2, e_3$, then they must be the same endpoint and thus $e_1,e_3$ share that as well. It remains to show that the clustering instance admits $0$-cost solutions iff the original instance admits a perfect 3D matching.
    \begin{itemize}
    \item Suppose this clustering instance admits a $0$-cost solution. Call this solution $S$. Fix a $t \in \set{1,2,3}$ and $v \in V_t$. Let $E(v)$ be the set of hyperedges that contain $v$. We know by construction that, $\forall e \in E(v), d_t(e,S) = 0$, i.e. $\exists e' \in S$ such that $e,e'$ agree on their vertex in $V_t$. But then this vertex must be $v$ itself. Thus $e' \in E(v) \cap S$, so $\abs{E(v) \cap S} \geq 1$.
    
    Since $\abs{V_t} = \abs S$, this means that for each $v \in V_t$, $\abs{E(v) \cap S} = 1$. So $S$ is a perfect 3D matching.
    \item Suppose $M$ is a perfect 3D matching in the original instance. Fix $t \in \set{1,2,3}$ and $e \in C$ such that $e \cap V_t = \set v$. We know that $M \cap E(v) \neq \emptyset$, so fix $e' \in M \cap E(v)$. Then $e \cap e' \cap V_t = \set v$, i.e. $d_t(e,e') = 0$, and so $d_t(e,M) = 0$.
    So for each $t \in \set{1,2,3}$ and each $e \in C$, we have $d_t(e,M) = 0$.
    Thus $\aggcost(M) = 0$. \qedhere
\end{itemize}
\end{proof}

Using techniques from graph drawing literature, one can show that our reduction \Cref{thm:T=3-hardness} holds even when the base graph is a grid graph. Furthermore, we can rule out $\poly(n)$-time algorithms for aggregate clustering whenever the base graph $G$ contains even a relatively small ($\tilde \Omega(\log^3 n) \times \tilde \Omega(\log^3 n)$) grid minor, and consequently has relatively small treewidth; this follows from the fact that, under the randomized exponential time hypothesis, it is hard to obtain a $2^{o(n)}$-time algorithm for 3D matching~\cite{KBI19, Kus20}. We state this formally below, and provide the proof in \Cref{sec:treewidth-lb-proof}.

\begin{restatable}{theorem}{treewidthLB}
\label{thm:treewidth-lb}
    Let $G$ be any $n$-vertex graph with treewidth $\tw = \Omega(\log^{295} n)$. (For example, one may simply take $G$ to be a $\sqrt{n} \times \sqrt{n}$ grid graph.) 
    For any homogenous aggregator $\Psi$ and any $z \in \NN \cup \set{\infty}$, 
    there is no $f(T)\poly(n,k)$-time finite-factor approximation for $\Psi$-aggregate $(k,z)$-clustering problem when the base graph inducing the metrics is $G$, unless the randomized exponential time hypothesis is false.
    This results holds even for $T=3$.
\end{restatable}

Next, we give a different reduction to rule out that FPT algorithms parameterized by $k$, even when the base graph is a star.

\begin{theorem}[Hardness of approximation on stars in FPT time.]\label{thm:w2-hardness}
    For any homogenous aggregator $\Psi$ and any $z \in \NN \cup \set{\infty}$, there is no $f(k)\poly(n,T)$-time finite-factor approximation algorithms for the $\Psi$-aggregate $(k,z)$-clustering problem 
  	even when the base graph $G$ inducing the metrics is a star, unless W[2] = FPT.
\end{theorem}

\begin{proof}
    We reduce from $k$-hitting set, which is W[2] hard when parameterized by $k$~\cite{CyganFKLMPPS2015}.
    Suppose we are given an instance of $k$-hitting set defined by a universe $U = \set{v_1, \ldots, v_n}$ of $n$ points, a set $\calX = \set{X_1, \ldots, X_m}$ of $m$ subsets of $U$, and an integer $k$; our task is to determine whether there are $k$ points in $U$ that hit all sets of $\calX$. We construct an instance of aggregate clustering as follows. Let $G$ be a star graph with a root vertex $r$ connected to $n$ leaves $v_1, \ldots, v_n$. Let the facility set be $F \coloneqq \set{v_1, \ldots, v_n}$, and let the client set (in every metric) be $C = \set{r}$. We define $T \coloneqq m$ distances functions, one for each set $X_t \in \calX$, as follows: in the $t$-th metric, we set the weight of edge $(r, v_i)$ to be 0 if $v_i \in X_t$ and $\infty$ otherwise. 

    It remains to show that the $\Psi$-aggregate $(k,z)$-clustering problem has finite cost iff there is a hitting set of size $k$.
    \begin{itemize}
        \item First we show that, if there is a hitting set $S$ of size $k$, then the $\aggcost(S) = 0$. Indeed, for every $t \in [T]$, there exists some $v_i \in S$ with $v_i \in X_t$ and so $\dist_t(r, S) \le \dist_t(r, v_i) = 0$. In particular, the $(k,z)$-clustering cost paid by the client set $C = \set{r}$ in each metric is $0$, and so the $\Psi$-aggregate cost is also 0 because $\Psi$ is a norm.
        \item Similarly, if there is a set $S$ of $k$ facilities such that $\aggcost(S) < \infty$, then we claim $S$ is a hitting set. Indeed, consider some $X_t \in \calX$. Since $S$ has finite cost, we know that $\dist_t(r, S) < \infty$ and in particular there exists some $v_i \in S$ with $\dist_t(r, v_i) = 0$. By construction, $v_i$ hits $X_t$.\qedhere
    \end{itemize}
\end{proof}

\begin{remark}
    The above reduction additionally shows that designing \emph{bicriteria} approximation algorithms for aggregate clustering remains hard, due to the hardness of approximation of hitting set \cite{Feige1998} --- there exists some $\beta = O(\log n)$ such that it is NP-hard to distinguish between the case that there is a hitting set of size $k$ and the case that there is no hitting set of size $\beta k$.
\end{remark}

\subsection{\texorpdfstring{Proof of \Cref{thm:treewidth-lb}}{Proof of Lower Bound on Arbitrary Treewidth}}\label{sec:treewidth-lb-proof}

We massage our reduction from \Cref{thm:T=3-hardness}. Observe that each of the $T=3$ metrics used in the reduction \Cref{thm:T=3-hardness} is a $0/\infty$ metric; that is, either points are at distance 0, or are at distance $\infty$. In particular, each metric can be realized as a disjoint union of paths (each with 0-weight edges), and so each metric is the shortest-path metric of a planar graph. However, it is not clear that there is an \emph{underlying base graph} $G$ that is planar, because vertices could belong to different paths in different metrics. We show that these sets of paths can be realized as reweightings of an underlying grid graph, or of any graph of sufficiently large treewidth; see \Cref{fig:grid-reweighting}.

This reweighting scheme essentially follows from a much more general result by Schaefer~\cite{schaefer2021new} about orthogonal drawings of planar graphs with fixed vertex locations. (Compare also to the classic theorem of Pach and Wenger~\cite{pach2001embedding} which proves something similar but without the guarantee of orthogonality.) 
\begin{lemma}[Theorem 9 of \cite{schaefer2021new}]
\label{lem:ortho-drawing}
    Let $G$ be a planar graph with maximum degree at most $4$, with $n$ vertices $x_1, \ldots, x_n$. Let $p_1, \ldots, p_n$ be a set of $n$ points in $\RR^2$. Then there is an orthogonal drawing of $G$ (ie, a drawing in which every edge is drawn as a combination of horizontal and vertical line segments) such that every vertex $v_i$ is located at $p_i$, and every edge has at most $10n + 6$ bends. Moreover, this drawing can be found in $O(n^2)$ time.
\end{lemma}
We adapt this lemma to our setting.
\begin{lemma}
\label{lem:grid-drawing}
    Let $n \in \NN$, and let $H$ be a sufficiently large $O(n^3)\times O(n^3)$ grid graph. There are $n$ vertices $X = \set{x_1, \ldots, x_n}$ on $H$ with the following property: for any $0/\infty$ metric $\dist$ on the vertices $X$,
    there is a way to assign edge weights to $H$ that induce a shortest-path metric $\dist_H$ such that, for any $x_i, x_j \in X$, we have $\dist(x_i, x_j) = \dist_H(x_i, x_j)$. Moreover, this assignment of weights can be found in polynomial time.
\end{lemma}
\begin{proof}
    Given an arbitrary $0/\infty$ metric $\dist$ on $X$, let $\set{X_1, \ldots X_k}$ denote a partition of $X$ such that any two vertices $x_i, x_j$ have $\dist(x_i, x_j) = 0$ iff they belong to the same part in the partition. Let $G$ be a graph that connects vertices within each $X_i$ by a path with 0-weight edges. Clearly $G$ is planar and has maximum degree at most 2, and $\dist_G(x_i, x_j) = \dist(x_i, x_j)$ for all $x_i, x_j \in X$.

    Let $\alpha = \Theta(n^2)$ with some sufficiently large hidden constant ($\alpha = 50 n^2$ would suffice). For every $i \in [n]$, define $p_i \in \RR^2$ by $p_i \coloneqq (i\alpha, i\alpha)$. By \Cref{lem:ortho-drawing}, there is an orthogonal drawing of $G$ such every vertex $x_i$ is drawn at location $p_i$. Now consider the set of ``bends'' in the drawing, ie the locations in $\RR^2$ at which the drawing of some edge switches between being horizontal and being vertical. Observe that there are at most $(10n+6)(3n-6) = O(n^2)$ bends in total, as there are at most $3n-6$ edges in a planar graph, and each edge has a drawing with at most $10n+6$ bends.  In particular, this means that for any two consecutive vertices $x_i$ and $x_{i+1}$, there are at most most $O(n^2)$ bends with $x$-coordinate (resp. $y$-coordinate) between the $x$-coordinates (resp. $y$-coordinates) of $x_i$ and $x_{i+1}$. As the $x$- (resp. $y$-) coordinates of $x_i$ and $x_{i+1}$ differ by $\alpha = \Theta(n^2)$, we may assume WLOG that the location of every bend has integer coordinates.
    
    In other words, we may assume WLOG the drawing of $G$ only uses edges of the $O(n\alpha) \times O(n \alpha)$ integer grid.

    This drawing tells us how to reweight the grid graph $H$. To be precise, we assume that $H$ is an $O(n\alpha) \times O(n \alpha)$ grid graph, and we let the vertex $(i \alpha, i \alpha) \in V(H)$ be the ``designated location'' for vertex $x_i \in X$.
    By our discussion above, any planar graph $G$ can be drawn using edges of $H$. We define a weight function on $E(H)$, where edge $e \in E(H)$ has weight 0 if it is used in the drawing, and weight $\infty$ otherwise.
    Clearly $\dist_G(x_i, x_j) = \dist_H(x_i, x_j)$ for all $x_i, x_j \in X$ as desired.
\end{proof}

\begin{figure}
    \centering
    \includegraphics[width=0.95\linewidth]{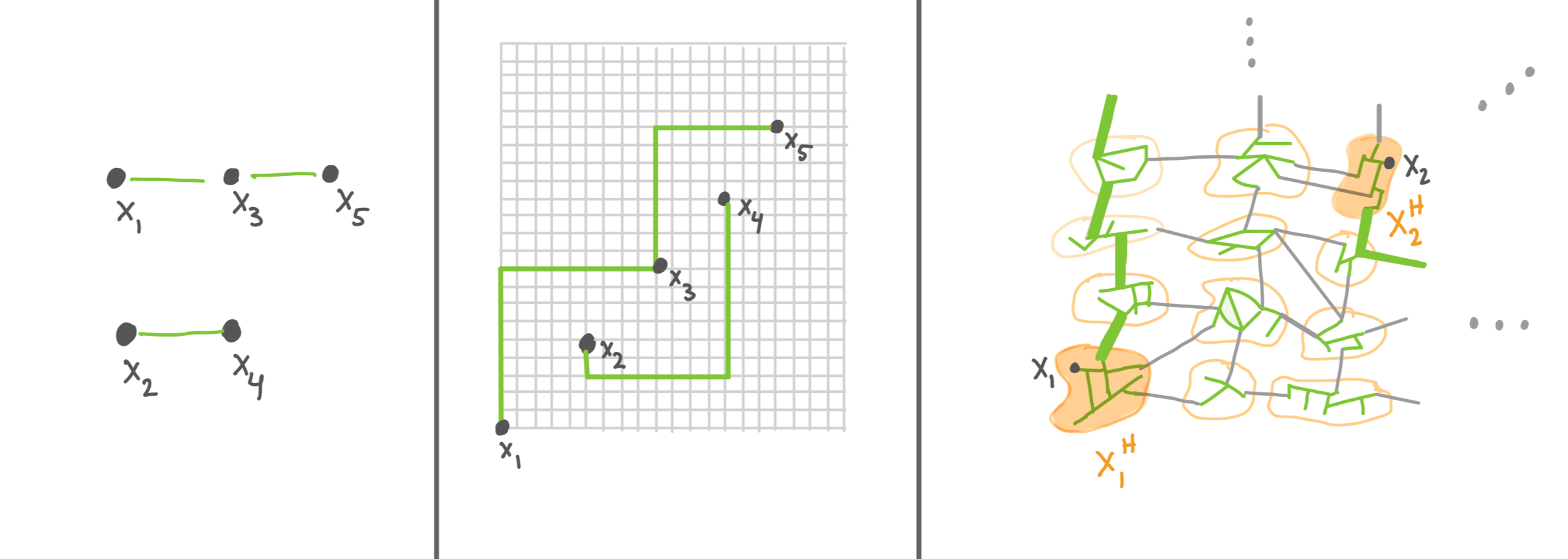}
    \caption{Left: A $0/\infty$-metric $d$ represented as a planar graph. Center: A reweighting of a grid graph $H$ that realizes $d$, as in \Cref{lem:grid-drawing}. Right: A reweighting of a large-treewidth graph $G$ that realizes $d$, as in \Cref{cor:drawing-on-treewidth}; the supernodes of a grid minor are drawn in orange. In all three images, green lines represent 0-weight edges, and gray lines represent $\infty$-weight edges}
    \label{fig:grid-reweighting}
\end{figure}

We observe that a similar reweighting can be carried out on any graph of high treewidth, because of the Excluded Grid theorem that connects treewidth to grid minors. The first version of the Excluded Grid theorem was proven by Robertson and Seymour \cite{robertson1986graph} but with a much worse relationship between the treewidth and grid size; after a series of improvements, polynomial bounds were finally given by Chekuri and Chuzhoy \cite{chekuri2016polynomial}. (Note that there has been subsequent work that improves the polynomial \cite{chuzhoy2015excluded,chuzhoy2016improved,chuzhoy2021towards}, but their results are non-constructive.)
\begin{lemma}[Excluded Grid theorem \cite{chekuri2016polynomial}]
    If $G$ is a graph with treewidth $\tw$, then $G$ contains a $(g\times g)$-grid as a minor for some $g = \Omega(\tw^{1/98}/\poly(\log \tw))$. Moreover, there is a randomized polynomial-time algorithm that finds the grid minor with high probability.
\end{lemma}

\begin{corollary}
\label{cor:drawing-on-treewidth}
    Let $n \in \NN$, and let $G$ be a graph with treewidth $\tw = \Omega(n^{294} \cdot \poly(\log n))$ for some sufficiently large $\poly(\log n)$. There are $n$ vertices $X = \set{x_1, \ldots, x_n}$ on $G$ with the following property: for any $0/\infty$ metric $\dist(\cdot, \cdot)$ on the vertices $X$,
    there is a way to assign edge weights to $G$ that induce a shortest-path metric $\dist_G$ such that, for any $x_i, x_j \in X$, we have $\dist(x_i, x_j) = \dist_H(x_i, x_j)$. Moreover, there is a randomized $\poly(|V(G)|)$-time algorithm that finds these edge weights with high probability.
\end{corollary}
\begin{proof}
    By the Excluded Grid theorem, $G$ contains a $(g\times g)$-grid minor $H$ for $g = \Omega(n^3)$. By definition, there is a minor model of $H$ in $G$; that is, there is a set $\calS$ of disjoint \emph{supernodes} (where each supernode is a connected subset of $G$) and a bijection $\psi:V(G) \to \calS$ such that for any vertices in the grid $a,b \in V(H)$, the existence of an edge $(a,b) \in E(H)$ in the grid implies that there is some edge in $G$ between a vertex in $\psi(a)$ and a vertex in $\psi(b)$.  Moreover, there is a randomized algorithm that finds this minor model with high probability, in $\poly(|V(G)|)$ time.

    By \Cref{lem:grid-drawing}, we can select $n$ designated vertices $X^H \subseteq V(H)$ such that any $0/\infty$ metric on $X^H$ can be realized by reweighting $H$. Define $X \subseteq V(G)$ by selecting, for every $x_i^H \in X^H$, an arbitrary vertex $x_i$ from the supernode $\psi(x_i^H)$.
    
    Now suppose we are given some $0/\infty$ metric $\dist$ on $X$. Equivalently, $\dist$ could be viewed as a metric on $X^H$; under this view, let $w_H:E(H) \to \RR$ denote the reweighting of $H$ that realizes $\dist$, as guaranteed by \Cref{lem:grid-drawing}. We define a reweighting $w:E(G) \to \RR$ as follows. For any edge $e \in E(G)$, if $e$ has both endpoints in the same supernode then $w(e) \coloneqq 0$. If $e$ has one endpoint in a supernode $\psi(a)$ and another in a supernode $\psi(b)$ for some $a,b \in V(H)$, and furthermore $(a,b) \in E(H)$, then set $w(e) \coloneqq w_H((a,b))$. In all other cases, set $w(e) \coloneqq \infty$. This reweighting satisfies $\dist_G(x_i, x_j) = \dist_H(x_i^H, x_j^H) = \dist(x_i, x_j)$ for all $x_i, x_j \in X$, as desired.
\end{proof}

For the proof of \Cref{thm:treewidth-lb}, we need one final ingredient, which is a refined perspective on the hardness of 3D matching. Recall that the \emph{exponential time hypothesis} (ETH) is a standard assumption in fine-grained complexity that asserts there does not exist any algorithm that decides 3SAT instances with $n$ literals in time $2^{o(n)}$.  The \emph{randomized exponential hypothesis} (rETH) asserts that this hardness holds even for randomized algorithms that decide 3SAT with high probability.
\begin{lemma}[\cite{KBI19, Kus20}]
\label{lem:eth}
    If rETH holds, then there is no randomized algorithm that solves perfect 3D matching with high probability in time $2^{o(m)}$, where $m$ is the number of hyperedges in the 3D matching instance.
\end{lemma}

We are finally ready to prove \Cref{thm:treewidth-lb}.
\begin{proof}[Proof of \Cref{thm:treewidth-lb}]
    Let $G$ be any graph with treewidth $\tw = \Omega(\log^{295} n)$.
    Let $(V_1 \uplus V_2 \uplus V_3, E)$ be an instance of perfect 3D matching with $m = \Theta(\log^{1.001} n)$ hyperedges. By \Cref{lem:eth}, it is rETH-hard to solve this instance in $\poly(n)$ time. As in \Cref{thm:T=3-hardness}, we define $T=3$ metrics (each a $0/\infty$ metric) such that the aggregate clustering problem encode the 3D matching instance. By \Cref{cor:drawing-on-treewidth}, each of these $T$ metrics is the shortest path metric of some reweighting of $G$, and these reweightings can be found (with high probability) in $\poly(n)$ time. We conclude that aggregate clustering is hard even when the underlying graph is $G$.
\end{proof}

\section{\texorpdfstring{Constant-factor approximations for $T=2$ scenarios}{Constant-factor approximations for T=2 scenarios}}
\label{sec:T=2}

In this section we extend the vanilla $O(1)$-approximations for $k$-supplier and $k$-median for the case of $T=2$ scenarios; this contrasts with the inapproximability for $T\geq 3$ (\Cref{thm:T=3-hardness}).
For simplicity, we restrict our exposition when the aggregator function $\Psi$ is just the sum, but we later point out in \Cref{rem:18} why it holds for any homogeneous aggregator such as a norm. 
In~\Cref{rem:19} we explain how our algorithms generalize for $(k,z)$-clustering.

At a high level, both the extensions from vanilla work because (a very simple) matroid {\em intersection} is polynomial time tractable, while even 3D-matching is NP-hard; the latter was the root of
the hardness for $T\geq 3$ aggregate clustering. We begin with the $k$-supplier problem which explains the previous line, and then show how ideas from the ``matroid-median'' problems solves the aggregate $k$-median problem for $T=2$.

\subsection{\texorpdfstring{$3$-approximation for Aggregate $k$-Supplier}{3-approximation for Aggregate k-Supplier}}\label{sec:T=2:ksupp}

Let $\OPT$ be the optimal set of $k$-suppliers which minimizes $\cost_\infty(d_1;\OPT) + \cost_\infty(d_2;\OPT)$, where recall $\cost_\infty(d_t;\OPT) = \max_{j\in C} d_t(j,\OPT)$.
Let us denote $\cost_\infty(d_t;\OPT)$ as $\opt_t$ for $t\in \{1,2\}$. Note that this can be guessed in $\poly(n)$-time (one could make this poly-logarithmic in $n$ using ``binary search'' techniques)
and so we assume we know $\opt_1$ and $\opt_2$. We now describe an algorithm, very similar to the one in~\cite{HochbS1986}, which returns a subset $\ALG$ of facilities
with $\cost_\infty(d_t;\ALG) \leq 3\opt_t$ for $t\in \{1,2\}$. This implies a $3$-approximation for the aggregate $k$-supplier problem for the $T=2$ case for any homogeneous aggregator $\Psi$. 

For $t\in \{1,2\}$, define $B_t(a,r) = \set{b \in F\cup C: d_t(a,b) \leq r}$. We run the Hochbaum-Shmoys~\cite{HochbS1986} filtering algorithm on each metric $d_t$ with parameter $\opt_t$.
To briefly describe this: (i) initially all clients are ``uncovered'', (ii) pick an arbitrary {\em uncovered client} $j\in C$ and add it to a set $R_t$ of representatives, define 
$B_t(j, 2\opt_t) \cap C$ to be $j$'s ``children'', and mark $j$ and all its children as ``covered'', and (iii) continue the above till all clients covered. The main observations, which are easy to check using triangle inequality and the fact that $\opt_t$'s were correct guesses, are the following for $t\in \{1,2\}$:
\begin{itemize}
	\item The collection of balls $\set{B_t(j,\opt_t) \cap F}_{j\in R_t}$ are pairwise disjoint.
	\item For any $j \in R_t$, $B_t(j, \opt_t) \cap \OPT \neq \emptyset$. The above two imply $|R_t| \leq k$.
	\item For every $v \in C$, $d_t(v,R_t) \leq 2\opt_t$.
\end{itemize}

\noindent
We now construct a solution $\ALG$ of size $\leq k$ such that for any $t\in \{1,2\}$ and any $j\in R_t$, we have $\ALG \cap B_t(j, \opt_t) \neq \emptyset$;
that is, $\ALG$ hits every ball in the above two collections. This can be solved using {\em matroid intersection} since there are only two collections of disjoint balls.
Once we have such an $\ALG$, it is easy to see using triangle inequality that, for any $t\in \{1,2\}$, $d_t(v, \ALG) \leq d_t(v, j) + d_t(j, \ALG) \leq 3\opt_t$ where $j$ was $v$'s representative in $R_t$.

To give more details on how to find $\ALG$, for $t\in \{1,2\}$ define the partition $\calP_t$ of $F$ formed by $\set{B_t(j,\opt_t) \cap F}_{j\in R_t} \cup Z_t$, where $Z_t$ are all the facilities of $F$ not in any ball.
We define the ``budget'' of each part to be $|\set{B_t(j,\opt_t) \cap F}| - 1$ --- this is the number of facilities $\ALG$ can ``leave out'' --- and set the budget of $Z_t$ to be $|Z_t|$.
We say $S\subseteq F$ is independent in partition matroid $t$ if it picks at most budget from each part; we seek the set $S\subseteq F$ of the largest cardinality which is independent in both partition matroids.
We return $\ALG := F\setminus S$. 
By design, $\ALG$ hits every $B_t(j,\opt_t)$ for $t\in \{1,2\}$ and $j\in R_t$. Since $F\setminus \OPT$ is a candidate $S$ which is in the intersection of the two partition matroids, 
we get $|\ALG| = |F| - |S| \leq |F| - |F\setminus \OPT| = |\OPT| \leq k$. 

In sum, we can find a solution $\ALG$ such that $\cost_\infty(d_t;\ALG) \leq 3\opt_t$ for $t\in \{1,2\}$, and thus $\Psi(\cost_\infty(d_1;\ALG), \cost_\infty(d_2;\ALG))\leq 3\Psi(\opt_1, \opt_2)$.
 This completes the proof of the following theorem.

\begin{theorem}\label{thm:k-supp:3appx}
    When $T=2$, there is a $\poly(n)$-time $3$-approximation algorithm for $\Psi$-aggregate $k$-supplier on $n$ vertices for any homogeneous aggregator $\Psi$ . 
\end{theorem}

\subsection{\texorpdfstring{$O(1)$-approximation for Aggregate $k$-Median}{O(1)-approximation for Aggregate k-Median}} \label{sec:k-med}

To obtain an constant approximation for the $k$-median problem one has to work a bit harder, but the underlying idea behind tractability is still 
the tractability of matroid intersection. More precisely, it is the integrality of a polytope defined by two laminar set systems. 
At a high level, such an issue arises when one studies the {\em matroid median} problem where there is only one metric but one has the extra constraint that 
the set of facilities opened must be an independent set of a certain matroid. One solves this problem by rounding a solution to an LP-relaxation where matroid intersection (or rather the integrality fact mentioned above) 
forms a core component with the given matroid being one matroid, and the other formed via ``filtering'' technique a la~\cite{ShmoysTA97, ChariGTS2002}. In our case the situation is similar at such a high-level -- the two ``matroids'' are formed by the filtering ideas for the $T=2$ scenarios -- but the details do need working out, and we show this below. 
We use the framework set by Swamy~\cite{Swamy2016} for the matroid-median problem, but other frameworks (such as the iterated rounding framework of~\cite{KrishLS2018}) could possibly lead to better approximation factors; in this work, we didn't optimize the latter. \smallskip

\paragraph*{Linear Programming Relaxation}

We begin by writing a linear programming relaxation for our problem. Recall that, in the standard linear programs for $k$-median, matroid median, etc.~\cite{ChariGTS2002,Li2013,Swamy2016,DengLR2022},
variables of the form $x(i,v)$ denote whether or not the client $v$ uses the facility $i$, so that $v$'s share of the cost is $\sum_{i \in F}d(v,i)x(i,v)$. In an integral solution $S \subseteq F$, each $v$ has a unique $i_v \in S$ such that $x(i_v,v) = 1$, and $d(i_v,v) = d(v,S)$.

In our problem, we need \emph{two} sets of such $x$-variables, because given an integral solution $S \subseteq F$, a client $v$ can use different facilities under different metrics; that is, there can be distinct $i_1,i_2 \in S$ s.t. $d_1(v,S) = d_1(v,i_1)$ and $d_2(v,S) = d_2(v,i_2)$. So for each $t \in \set{1,2}$, we define variables $\set{x_t(i,v)}_{v \in C, i \in F}$ denoting whether or not the client $v$ uses the facility $i$ under the metric $d_t$. We also have variables $\set{y(i)}_{i \in F}$ which, like in the standard $k$-median LP~\cite{ChariGTS2002,ChariL2012}, denote whether or not the facility $i$ is picked into our solution. This allows us to write linear constraints similar to the standard LP for $k$-median. We use $\lp$ to refer to the fractional optimum of the linear program relaxation.
\begin{alignat}{5}
	\text{minimize:} & \sum_{v \in C}\sum_{i \in F}d_1(i,v)x_1(i,v) + \sum_{v \in C}\sum_{i \in F}d_2(i,v)x_2(i,v)\tag{\LP}\label{kmed-cp:obj}\\
	&\sum_{i \in F}y(i) \leq k\tag{\LP1}\label{kmed-cp:1-k}\\
	&x_t(i,v) \leq y(i) &&\forall t \in \set{1,2}, v \in C, i \in F\tag{\LP2}\label{kmed-cp:2-san}\\
	&\sum_{i \in F}x_t(i,v) = 1 &&\forall t \in \set{1,2}, v \in C\tag{\LP3}\label{kmed-cp:3-cov}\\
	&0 \leq x_t(i,v), y(i) \leq 1 &&\forall t \in \set{1,2}, v \in C, i \in F\notag
\end{alignat}
\noindent

Given a solution $(x_1, x_2, y)$ of \ref{kmed-cp:obj}, we first make the following assumption which follows from {\em facility splitting} arguments due to Chudak and Shmoys~\cite{ChudaS2003} (see also~\cite{YanC2015})
\begin{assumption}\label{assn:completeness}
	$\forall t \in \set{1,2}, v \in C, i \in F$, $x_t(i,v) \in \set{0,y(i)}$.    
\end{assumption}

We now perform a {\em filtering step} inherent to almost all LP-rounding algorithms for $k$-median. For our problem, as in the $k$-supplier problem from previous section, we get {\em two} sets of ``representatives'' instead of one that $k$-median rounding algorithms get. \smallskip

\noindent 
\textit{\bf Filtering.} For $t\in \set{1,2}$ and any client $v\in C$, define $C_t(v) := \sum_{i \in F}d_t(i,v)x_t(i,v)$, i.e. the cost paid by $v$ under metric $d_t$. Our filtering algorithm begins with all clients in an \emph{uncovered} set $U$. We pick $j \in U$ with the smallest $C_t(j)$, and call it a \emph{representative under $d_t$}, adding it to the set $\Reps_t$. Every \emph{uncovered} client $v$ that is within distance $4C_t(v)$ of $j$ becomes a \emph{child of $j$ under $d_t$}, forming the set $\child_t(j)$. All of $\child_t(j)$ is then considered covered. We repeat this until all clients are covered. So we get $\Reps_t \subseteq C$ that is \emph{well-separated}, i.e. any two distinct $j,j' \in \Reps_t$ have $d_t(j,j') > 4\max\set{C_t(j),C_t(j')}$. $\Reps_t$ also induces a partition $\set{\child_t(j)}_{j \in \Reps_t}$ of $C$. The following lemma is standard (see \Cref{sec:half-integ-details} for a proof), and thereafter, we focus on $\Reps_1$ and $\Reps_2$, and seek a constant-factor approximation on those clients only.

\begin{restatable}{lemma}{lemfiltering}
	\label{lem:filtering-4}
	Consider $S \subseteq F$ and $\alpha \geq 1$, such that $\forall t \in \set{1,2}$,\\ $\sum_{j \in \Reps_t}\abs{\child_t(j)}d_t(j,S) \leq \alpha \cdot \sum_{v \in C}C_t(v)$.
	Then $\aggcost(S) \leq (4+\alpha) \cdot \opt$.
\end{restatable}
\noindent
The rounding algorithms for many $k$-median algorithms first round to a half-integral solution and then to integral. We follow the same route, and in particular, follow Swamy's framework~\cite{Swamy2016}.
We begin with the following definitions for each $t \in \set{1,2}, j \in \Reps_t$; as in the previous section, we use $B_t(a,r) = \set{b\in F\cup C : d_t(a,b) \leq r}$. An illustrative figure for these appears in Swamy's work~\cite[Figure 1]{Swamy2016}.
\begin{itemize}
	\item $F_j^t := \set{i \in F : d_t(i,j) = \min_{j' \in \Reps_t}d_t(i,j')}$ (breaking ties arbitrarily)
	\item $B_j^t := B_t(j,2C_t(j)) \cap F$. Notice that, by construction of $\Reps_t$, $B_j^t \subseteq F_j^t$ (also follows from \Cref{fact:swamy:laminar}).
	\item $\gamma_j^{\paren t} := \min_{i \notin F_j^t}d_t(i,j)$, and $G_j^t := \set{i \in F_j^t : d_t(i,j) \leq \gamma_j^{\paren t}}$.
\end{itemize}
We obtain the following a la~\cite{Swamy2016} (see \Cref{sec:half-integ-details} for a proof).
\begin{restatable}{fact}{factswamy}
	\label{fact:swamy}
	For every $t \in \set{1,2}, j \in \Reps_t$, we have
	\begin{itemize}
		\item $y(B_j^t) := \sum_{i \in B_j^t}y(i) \geq \half$ \label{fact:swamy:half}
		\item $B_j^t \subseteq G_j^t$\label{fact:swamy:laminar}
		\item Suppose $\gamma_j^{\paren t} = d_t(i_\gamma,j)$, s.t. $i_\gamma \in F^t_\ell$. Then $\forall i \in B_\ell^t, d_t(i,j) \leq 3\gamma_j^{\paren t}$.\label{fact:swamy:gamma}
	\end{itemize}
\end{restatable}

\noindent
Also note that by design the $F_j^t$'s (and therefore the $G_j^t$'s) are pairwise disjoint for a fixed $t$ when we consider $j\in \Reps_t$.
These play the role of the ``partitions'' as in the $k$-supplier problem, and our tractability follows because we have only two $t$'s. 
Instead of finding the ``largest cardinality independent set'', as we did for the $k$-supplier problem, we instead find a point maximizing a suitable linear function.
Towards this, we describe a ``linear function'' that acts as a proxy for $C_t(j)$'s.

For fractional facility masses $z \in [0,1]^F$, $\forall t \in \set{1,2}$, and $\forall j \in \Reps_t$, let $z(G_j^t) := \sum_{i \in G_j^t}z(i)$. Then we have
\begin{align*}
	&T_t(z,j) := \sum_{i \in G_j^t}d_t(i,j)z(i) + 3\gamma_j^{\paren t}\max\set{0,1-z(G_j^t)},\, \text{and}\\
	&T_t(z) := \sum_{j \in \Reps_t}\abs{\child_t(j)}T_t(z,j)\,.
\end{align*}
The above definition is quite similar to Swamy's definition which may be taken as the $T=1$ case, with one difference that we have the ``$\max$ with 0''. When $T=1$, one can assert $z(G_j) \leq 1$
but with different $t$'s, it may be that $z(G_j^2) > 1$ because in metric $1$ we need to open (fractionally) a lot of facilities. Nevertheless, even with the ``max'' function the following
two lemmas similar to those in~\cite{Swamy2016} hold. The full proofs are in \Cref{sec:half-integ-details}.

\begin{restatable}{lemma}{lemTproxy}
	\label{lem:T-proxy}
	Consider $z \in [0,1]^F$ s.t. $\forall t \in \set{1,2}, j \in \Reps_t$, $z(B_j^t) \geq \half$. Then $j$'s assignment cost under $d_t$ is at most $T_t(z,j)$. That is, for every $t \in \set{1,2}, j \in \Reps_t$, we can assign variables $\set{0 \leq x'_t(i,j) \leq z(i)}_{i \in F}$ so that $\sum_{i \in F}d_t(i,j)x'_t(i,j) \leq T_t(z,j)$.
\end{restatable}

\begin{restatable}{lemma}{lemTupperbd}\label{lem:T-upper-bd}
	For each $t \in \set{1,2}$, and each $j \in \Reps_t$, $T_t(y, j) \leq 3\cdot C_t(j)$.
\end{restatable}

\noindent
{\bf Rounding to Half-integral Solution.} We are now ready to describe the algorithm which is encapsulated in the following lemma. 

\begin{lemma}\label{lem:half-integral}
	There is a polynomial time algorithm that yields a half-integral solution $(\hat x_1, \hat x_2, \hat y)$ of \ref{kmed-cp:obj} s.t. $\forall t \in \set{1,2}$, $\sum_{j \in \Reps_t}\abs{\child_t(j)}\sum_{i \in F}d_t(i,j)\hat x_t(i,j) \leq 6 \cdot \sum_{v \in C}C_t(v)$.
\end{lemma}

Our goal now becomes to find a half-integral $\hat y \in \set{0,\half,1}^F$ so that, for each $t \in \set{1,2}$, $T_t(\hat y) \leq 6\cdot \sum_{v \in C}C_t(v)$. 
If we are able to do so, then by~\Cref{lem:T-proxy} we would be done.
To find such a half-intergral solution, we define the following polytope and linearization of $T_t$'s.
The polytope $\calP$ has variables $(z,\lambda_1,\lambda_2) \in \RR_{\geq 0}^{F \cup \Reps_1 \cup \Reps_2}$. The auxiliary $\lambda_t$ variables serve to linearize the $\max$ terms in the $T_t(z)$'s, as the constraints in $\calP$ ensure that $\lambda_t(j) \geq \max\set{0,1-z(G_j^t)}$.

\begin{align*}
	\calP := \Bigg\{(z,\lambda_1,\lambda_2) \in \RR_{\geq 0}^{F \cup \Reps_1 \cup \Reps_2} : & z(F) \leq k;\,\forall t \in \set{1,2},\forall j \in \Reps_t,\\
	&z(B_j^t) \geq \half,\,z(G_j^t) + \lambda_t(j) \geq 1\Bigg\}
\end{align*}
\noindent
The following follows from the presence of only ``two partitions'' plus noting the $\lambda$'s don't bother much since they are identity.
\begin{claim}\label{clm:P-half-integral}
	$\calP$ has half-integral extreme points.
\end{claim}
\begin{proof}
	Consider the constraint matrix $M$ of $\calP$. Since the constant terms in the constraints are  half-integral, it suffices~\cite{HoffmK1956} to show
	that $M$ is totally unimodular (TU). Indeed, $M$ looks 
  $  \left[
\begin{array}{c|c}
\multirow{2}{*}{$A$} & I \\
                   & 0
\end{array}
\right]$
	 where we have the $A$-matrix corresponding to the $z$-variables, and the $I$ corresponding to the $\lambda_t$-variables. 
     
	The dimension of $I$ is $|\Reps_1| + |\Reps_2|$. The matrix $A$ is the incidence matrix of two laminar systems --- in fact, it is two partitions
	coarsened by the universe. Such a system is TU~\cite{Swamy2016}. A TU matrix padded with columns with at most one $1$ in them remains TU.
\end{proof}
Over $\calP$, we can then consider the {\em minimization} of the following linearization of the $T_t(z)$'s: $\forall t \in \set{1,2}, \forall j \in \Reps_t$: $W_t(z,\lambda_t,j) := \sum_{i \in G_j^t}d_t(i,j)z(i) + 3\gamma_j^{\paren t}\lambda_t(j)$, and
\begin{align*}
	W_t(z,\lambda_t) := \sum_{j \in \Reps_t}\abs{\child_t(j)}W_t(z,\lambda_t, j)\,.
\end{align*}
The following claim encapsulates that we can find an extreme point of (any polytope) $\calP$ which is a $2$-approximation to both $W_t$'s {\em simultaneously} in a sense made clear below. 
This allows us to get the desired $\hat y$ (as we explain better in the proof of~\Cref{lem:half-integral}).
\begin{claim}\label{clm:T-sample}
	Consider a polytope $\calP \subseteq \RR_{\geq 0}^m$, linear functions $W_1, W_2 : \RR^m_{\geq 0} \to \RR_{\geq 0}$, and a point $p \in \calP$. 
	There is a polynomial time algorithm which returns an extreme point $\hat p$ of $\calP$ such that 
	$W_t(\hat p) \leq 2W_t(p)$ for $t\in \{1,2\}$.
\end{claim}
\begin{proof}
	By Carath\'{e}odory's theorem~\cite{Carat1911,Stein1913}, we can write $p$ as a linear combination of at most $(m+1)$ extreme points of $\calP$.
	More precisely, we can find, in polynomial time, a subset $\calE$ of extreme points such that $|\calE| \leq m+1$ and $p=\sum_{q \in \calE}\mu_q\cdot q$
	for some $\mu_q$'s which form a probability distribution $D_p$. We claim that one of the $q\in \calE$ is the desired point, and we can find which one by enumeration.
	To see why, observe that for $t\in \{1,2\}$,  $\EXP_{\hat p \sim D_p}[W_t(\hat p)] = W_t(p)$, and so by Markov, 
	$\Pr_{\hat p \sim D_p}[W_t(\hat p) > 2W_t(p)] < \frac{1}{2}$. So, by union bound $\Pr_{\hat p \sim D_p}[\exists t\in \{1,2\}:~W_t(\hat p) > 2W_t(p)] < 1$ implying
	$\Pr_{\hat p \sim D_p}[W_t(\hat p) \leq 2W_t(p), ~~t\in \{1,2\}] > 0$ which proves the above claim.
\end{proof}
\begin{proof}[Proof of \Cref{lem:half-integral}]
	Given our initial solution $(x_1,x_2,y)$ of \ref{kmed-cp:obj}, set variables $\lambda_t(j) := \max\set{0,1-y(G_j^t)}$ for each $t \in \set{1,2}$, $j \in \Reps_t$. By \Cref{fact:swamy}, this gives us $(y,\lambda_1,\lambda_2) \in \calP$. So by \Cref{clm:P-half-integral} and \Cref{clm:T-sample}, we can obtain a half-integral $(\hat y, \hat \lambda_1, \hat \lambda_2) \in \calP$ such that $W_1(\hat y,\hat\lambda_1) \leq 2W_1(y,\lambda_1)$, and $W_2(\hat y,\hat \lambda_2) \leq 2W_2(y,\lambda_2)$.
	So by \Cref{lem:T-proxy}, for each $t \in \set{1,2}$ we can construct $\hat x_t$'s such that,
	\begin{align*}
		\sum_{j \in \Reps_t}\abs{\child_t(j)}\sum_{i \in F}d_t(i,j)\hat x_t(i,j) & \leq T_t(\hat y) \leq W_t(\hat y,\hat\lambda_t) \leq 2W_t(y,\lambda_t)
	\end{align*}
	where the penultimate step follows from the constraints in $\calP$. By our construction of the starting $\lambda_t$'s, i.e. the constraints enforcing $\lambda_t(j) \ge \max\{0,1-z(G_j^t)\}$, $W_t(y,\lambda_t) = T_t(y)$. So we have, by \Cref{lem:T-upper-bd}, that the above can be upper bounded by
	$
	6 \sum_{j \in \Reps_t}\abs{\child_t(j)}C_t(j) \leq 6\sum_{v \in C}C_t(v)
	$.
\end{proof}

\noindent
\textbf{{Rounding to Integral solution.}} We round $(\hat x_1, \hat x_2, \hat y)$ to an integral solution, losing a constant factor, and this is very similar to the ideas in~\cite{ChariGTS2002, Swamy2016}; here also we follow the latter's framework. 

Below are the main ideas, and the details are in \Cref{sec:integ-proofs}.

The first step is another {\em filtering} step. For $t\in \{1,2\}$ and $j\in \Reps_t$, define $S^t_j := \{i\in F~:~\hat x_t(i,j) > 0\}$ and let $\hat C_t(j) := \sum_{i\in F} \hat x_t(i,j)d_t(i,j)$;
note that $\sum_{j\in \Reps_t}|\child_t(j)|\hat C_t(j) \leq 6\sum_{v\in C} C_t(v)$. Now we define a subset of ``super-representatives'': we pick $\ell \in \Reps_t$ with smallest $\hat C_t(\ell)$
and add it to $\Super_t$ removing every $j$ with $S^t_j \cap S^t_\ell \neq \emptyset$ from consideration. In the end, we have $S^t_j$'s pairwise disjoint for $\Super_t$'s and each $j\in \Reps_t$ shares
a facility with a super-representative in $\Super_t$. 

Next, as in the rounding to the half-integral case, we define a particular linear function which works as a proxy for $\hat C_t(j)$. To define this proxy, for every $j\in \Reps_t$ call the two facilities in $S^t_j$
primary and secondary in ascending order of distance (duplicating facilities if needed). For $t\in \{1,2\}$ and $j\in \Reps_t$, let $\ell \in \Super_t$ be the super-representative which $j$ shares a facility with, and define
\[A_t(z,j) := \begin{cases}
	\sum_{i \in S_\ell^t}d_t(i,j)z(i) \hfill \text{if }\prim_t(j) \in S_\ell^t\\
	\sum_{i \in S_\ell^t}d_t(i,j)z(i) + \paren{d_t(\prim_t(j),j) - d_t(\sec_t(j), j)}\cdot z(\prim_t(j)) &\text{otherwise}
\end{cases}\]
It's instructive to think of all $z(i)$'s as $1/2$; in that case $A_t(z,j)$ is either the average of $j$'s distance to $\ell$'s facilities when $j$'s closest facility is shared with $\ell$, 
or it is the average of $j$'s distance to it closest facility and distance to the other facility that $\ell$ goes to whom $j$ doesn't share. The importance of the above function is captured in a couple of observations: (A) for any $t$ and any $j\in \Reps_t$, we can upper bound $A_t(\hat y, j) \leq 2\hat C_t(j)$, that is, $A_t(\hat y, j)$ isn't too bad a proxy, and perhaps more usefully (B) given any $z$ such that $z(S^t_\ell) = 1$ for all {\em super representatives}, we can find an assignment of any $j\in \Reps_t$ to facilities with connection cost $\leq A_t(z, j)$. These two facts, and the fact that we can minimize linear functions over the intersection of two partition matroids (defined by the $S^t_\ell$'s for $t\in \{1,2\}$ and $\ell \in \Super_t$), gives us an $O(1)$-approximation. This is because the relevant polytope has integral extreme points (see \Cref{obs:R-integral}).
Observation (A) holds because we can charge $j$'s journey to $\ell$'s ``other facility'' to $\ell$'s connection cost which is smaller than $j$'s cost by design; this uses half-intergrality. The details are in~\Cref{lem:25}. Observation (B) is immediate if $\prim_t(j) \in S_\ell^t$ or if $z(\prim_t(j)) = 0$; otherwise, if $z(\prim_t(j)) = 1$ then $j$ will travel to its closest facility while $A_t(z,j)$ 
will be average of two number which are larger than the minimum. The statement is in fact true for fractional $z$'s as well and the proof is in~\Cref{lem:24}.

All in all, we get the following theorem; the factor $28$ could 
definitely be improved, but perhaps no better than $8$ using these methods.

\begin{theorem}\label{thm:k-med:const-appx}
    When $T=2$, there is a polynomial time $28$-approximation algorithm for sum-aggregate $k$-median.
\end{theorem}

\begin{remark}[Generalizing to arbitrary norm aggregators]\label{rem:18}
	In the above proof, we note that the solution $S$ we return has the property that $\cost(d_t;S) \leq 28\sum_{j\in C} C_t(j)$, for $t\in \{1,2\}$.
	This allows us to generalize the above theorem for any aggregator $\Psi$ which is a norm. Instead of a linear program, we would have a {\em convex} program.
	More precisely, we have variables $C_1, C_2$ where $C_t := \sum_{j\in C} C_t(j)$ and the objective would minimize $\Psi(C_1, C_2)$. 
	By the property of our rounding and the homogeneity of $\Psi$, we would get that $\Psi(\cost(d_1;S), \cost(d_2;S)) \leq 28\Psi(C_1, C_2) \leq 28\opt$.
\end{remark}

\begin{remark}[Generalizing $(k,z)$-clustering]\label{rem:19}
	The above theorem focused on the $k$-median problem. However, the same methodology also gives an $O(1)$-approximation for the $\Psi$-aggregated $(k,z)$-clustering problem
	due to the fact that we only use triangle inequalities over ``bounded number of hops''. This is a folklore observation (see, for instance, Footnote 1 in~\cite{ChakrS19}).
	To obtain this, first we replace $d_t(i,j)$ with $d_t(i,j)^z$ in the linear/convex program. 
	Next, we use the ``relaxed triangle inequality'', which follows from Lemma A.1 in~\cite{MakarMR19}, that says if we have $r+1$ points $a_1, \ldots, a_{r+1}$, then $d(a_1, a_{r+1})^z \leq r^{z-1} \sum_{i=1}^{r} d(a_i, a_{i+1})^z$.
		In our proof above, we never invoke the triangle inequality on more than $4$ points, and thus, everything goes through with a ``hit'' of $3^{z-1}$.
		Note, though, that the definition of the proxy function $T_t(z,j)$ would have $3$ replaced by $3^z$. 
		In the end, we would get a $O(1)^z$ approximation to the sum of the $z$th powers, and since we take the $z$th root, this gives a $O(1)$-approximation. 
\end{remark}

\subsection{Deferred proofs towards half-integrality}\label{sec:half-integ-details}

\lemfiltering*
\begin{proof}
    For each $t \in \set{1,2}$, we have by construction that $d_t(v,j) \le 4C_t(j) \le 4 C_t(v)$. Using this, we have
    \begin{align*}
        \sum_{v \in C}d_t(v,S) &= \sum_{j \in \Reps_t}\sum_{v \in \child_t(j)}d_t(v,S) \leq \sum_{j \in \Reps_t}\sum_{v \in \child_t(j)} \paren{d_t(v,j) + d_t(j,S)}\\
        &\leq \sum_{j \in \Reps_t}\sum_{v \in \child_t(j)}(4C_t(v)) + \sum_{j \in \Reps_t}\abs{\child_t(j)}d_t(j,S) \leq (4+\alpha) \cdot \sum_{v \in C}C_t(v)\,.
    \end{align*}
    So we have
    \begin{align*}
        \aggcost(S) &= \sum_{v \in C}d_1(v,S) + \sum_{v \in C}d_2(v,S)\\
        &\leq (4+\alpha)\cdot \sum_{v \in C}C_1(v) + (4 + \alpha) \cdot \sum_{v \in C}C_2(v)\\
        &= (4+\alpha)\cdot\lp \leq (4+\alpha)\cdot\opt\,.
    \end{align*}
\end{proof}

\factswamy*
\begin{proof}
Fix $t \in \set{1,2}$ and $j \in \Reps_t$.
    \begin{itemize}
        \item If $C_t(j) = 0$, then $y(B_j^t) = 1$. Otherwise,
        \begin{align*}
            C_t(j) &\geq \sum_{i \notin B_j^t}x_t(i,j)d_t(i,j) > 2C_t(j)\cdot\paren{\sum_{i \notin B_j^t}x_t(i,j)}\\
            &= 2C_t(j)\cdot\paren{1-\sum_{i \in B_j^t}x_t(i,j)} \geq 2C_t(j) \cdot \paren{1-y(B_j^t)}
        \end{align*}
        which, when $C_t(j) > 0$, gives $y(B_j^t) \geq \half$.
        \item By construction of the $\child_t$ sets, $B_j^t \in F_j^t$, so it suffices to show that $\gamma_j^{\paren t} \geq 2C_t(j)$. For this, suppose $\gamma_j^{\paren t} = d_t(i_\gamma,j)$, $i_\gamma \in F_\ell^t$. Then $\gamma_j^{\paren t} \geq d_t(j,\ell) - d_t(i_\gamma,\ell) \geq d_t(j,\ell) - \gamma_j^{\paren t}$. So $2\gamma_j^{\paren t} \geq d_t(j,\ell) > 4C_t(j)$ by the construction of $\child_t$ sets; i.e. $\gamma_j^{\paren t} \geq 2C_t(j)$.
        \item We have $d_t(i,j) \leq d_t(j,\ell) + d_t(i,\ell)$. We also have, by the construction of the $\child_t$ sets, that that $d_t(i,\ell) \leq 2C_t(\ell) \leq 2\max\set{C_t(\ell), C_t(j)} \leq \half \cdot d_t(\ell,j)$. So $d_t(i,j) \leq \frac 3 2 \cdot d_t(j,\ell)$. Finally, since $i_\gamma \in F_\ell^t$, we have $d_t(i_\gamma,\ell) \leq d_t(i_\gamma,j)$, and hence $d_t(j,\ell) \leq d_t(i_\gamma,j) + d_t(i_\gamma,\ell) \leq 2d_t(i_\gamma,j) = 2\gamma_j^{\paren t}$. So $d_t(i,j) \leq 3\gamma_j^{\paren t}$.
    \end{itemize}
\end{proof}

\lemTproxy*
\begin{proof}
    If $z(G_j^t) \geq 1$, then we can set $x'_t(i,j)$'s such that $\sum_{i \in G_j^t}x'_t(i,j) = 1$ and\\ $\sum_{i \notin G_j^t}x'_t(i,j) = 0$. This gives the assignment cost
    \[\sum_{i \in F}d_t(i,j)x'_t(i,j) = \sum_{i \in G_j^t}d_t(i,j)x'_t(i,j) \leq \sum_{i \in G_j^t}d_t(i,j)z(i) \leq T_t(z,j)\,.\]
    Otherwise, we set $x'_t(i,j) = z(i)$ for every $i \in G_j^t$. Then, we identify $\ell \in \Reps_t$ s.t. $\gamma_j^{\paren t} = d_t(i',j)$ for some $i' \in F_\ell$. \Cref{fact:swamy} gives us that $z(B_\ell^t) \geq \half \geq 1-z(B_j^t) \geq 1-z(G_j^t) = 1-\sum_{i \in G_j^t}x'_t(i,j)$. So we can set $x'_t(i,j)$'s s.t. they are zero outside $G_j^t \cup B_\ell^t$, and $\sum_{i \in B_\ell^t}x'_t(i,j) = 1-\sum_{i \in G_j^t}z(i)$. This gives us the assignment cost
    \begin{align*}
        \sum_{i \in F}d_t(i,j)x'_t(i,j) &= \sum_{i \in G_j^t}d_t(i,j)x'_t(i,j) + \sum_{i \in B_{j'}^t}d_t(i,j)x'_t(i,j)\\
        &\leq \sum_{i \in G_j^t}d_t(i,j)z(i) + 3\gamma_j^{\paren t}\paren{\sum_{i \in B_{j'}^t}x'_t(i,j)}\\
        &= \sum_{i \in G_j^t}d_t(i,j)z(i) + 3\gamma_j^{\paren t}\paren{1-z(G_j^t)} = T_t(z,j)\,.
    \end{align*}
\end{proof}

\lemTupperbd*
\begin{proof}
If $y(G_j^t) \geq 1$, then we have that
\[C_t(j) = \sum_{j \in \Reps_t}x_t(i,j)d_t(i,j) = \sum_{j \in \Reps_t}d_t(i,j)y(i) = T_t(y,j)\,.\]
    Otherwise, we can say by \Cref{assn:completeness} that
    \begin{align*}
    C_t(j) &= \sum_{i \in G_j^t}d_t(i,j)x_t(i,j) + \sum_{i \notin G_j^t}d_t(i,j)x_t(i,j)\\
    &\geq \sum_{i \in G_j^t}d_t(i,j)y(i) + \gamma_j^{\paren t}\cdot\paren{1-\sum_{i \in G_j^t}x_t(i,j)} &\dots\text{by \eqref{kmed-cp:2-san}}\\
    &= \sum_{i \in G_j^t}d_t(i,j)y(i) + \gamma_j^{\paren t}\paren{1-y(G_j^t)} \geq \frac 1 3 \cdot T_t(y,j)\,.
    \end{align*}
\end{proof}

\subsection{Rounding to an integral solution}\label{sec:integ-proofs}

Given the half-integral solution $(\hat x_1,\hat x_2, \hat y)$, we first make the assumption below.
\begin{assumption}\label{assn:completeness-again}
    $\forall i \in F$, $\hat y(i) \in \set{0,\half}$.
\end{assumption}
This assumption can be made true by splitting any $i \in F$ s.t $\hat y(i) = 1$ into two copies $i',i''$ with $\hat y(i') = \hat y(i'') = \half$; and then reassigning $\hat x_t$'s accordingly. This then implies \Cref{assn:completeness} as well.

\textit{Filtering.} We now perform another simpler filtering step in each metric. For $t \in \set{1,2}$ and $j \in \Reps_t$, we define $S_j^t := \set{i \in F : \hat x_t(i,j) > 0}$, and $\hat C_t(j) := \sum_{i \in F}\hat x_t(i,j)d_t(i,j)$. We begin with all of $\Reps_t$ being in a set $U$ of \emph{uncovered} representatives. We pick $\ell \in U$ with the smallest $\hat C_t(\ell)$, and call it a \emph{super-representative under $d_t$}, adding it to the set $\Super_t$. We define its \emph{neighbor set} $\nbr_t(\ell) := \set{j \in \Reps_t : S_j^t \cap S_\ell^t \neq \emptyset}$. We consider $\nbr_t(\ell)$ covered, and then repeat this process until all of $\Reps_t$ is covered.

Unlike in the previous phase, we do not discard $\Reps_t$ after forming $\Super_t$. We open facilities near the super-representatives, but still consider the assignment cost of all representatives.

\textit{Rounding.} For a $t \in \set{1,2}$, and $j \in \Reps_t$, we define $\prim_t(j)$ to be the facility in $S_j^t$ that is closer to $j$, and $\sec_t(j)$ to be the other one. Now suppose $j \in \nbr_t(\ell)$ for an $\ell \in \Super_t$. We define, for a $z \in [0,1]^F$,
\[A_t(z,j) := \begin{cases}
    \sum_{i \in S_\ell^t}d_t(i,j)z(i) \hfill \text{if }\prim_t(j) \in S_\ell^t\\
    \sum_{i \in S_\ell^t}d_t(i,j)z(i) + \paren{d_t(\prim_t(j),j) - d_t(\sec_t(j), j)}\cdot z(\prim_t(j)) &\text{otherwise}
\end{cases}\]
and $A_t(z) = \sum_{j \in \Reps_t}\abs{\child_t(j)}A_t(z,j)$. $A_t(z)$ serves as a linear proxy for the cost under metric $d_t$. We observe, directly from \ref{kmed-cp:3-cov}, that

\begin{observation}
    $\forall t \in \set{1,2}, j \in \Reps_t$, $\hat y(S_j^t) = 1$.
\end{observation}

We then obtain bounds analogous to \Cref{lem:T-proxy,lem:T-upper-bd}.

\begin{lemma}\label{lem:A-proxy}\label{lem:24}
    Consider a $z \in [0,1]^F$ s.t. for each $t \in \set{1,2}$, and $\ell \in \Super_t$, we have $z(S_\ell^t) = 1$. Then the assignment cost of $j$ under $d_t$ is at most $A_t(z,j)$. That is, we can assign variables $\set{0\leq x'_t(i,j) \leq z(i)}_{i \in F}$ such that, for each $t \in \set{1,2}$ and $j \in \Reps_t$, $\sum_{i \in F}x'_t(i,j)d_t(i,j) \leq A_t(z,j)$.
\end{lemma}

\begin{proof}
    If $z(\prim_t(j)) = 0$, or if $\prim_t(j) \in S_\ell^t$, then $A_t(z,j) = \sum_{i \in S_\ell^t}d_t(i,j)z(i)$. In these cases, set $x'_t(i,j) = z(i)$ for each $i \in S_\ell^t$, and other $x'_t(i,j)$'s to zero.

    Now consider the case where $\prim_t(j) \notin S_\ell^t$, and $z(\prim_t(j)) > 0$. Set $x'_t(\prim_t(j),j) = z(\prim_t(j))$. By construction of $\Super_t$, we must have $\sec_t(j) \in S_\ell^t$. Let the other facility in $S_\ell^t$ be $i_0$. Note that since $\hat x_t(i_0,j) < \hat x_t(\sec_t(j))$, we can assume that $d_t(\sec_t(j),j) \leq d_t(i_0,j)$; because if not, then swapping the values of $\hat x_t(i_0,j)$ and $\hat x_t(\sec_t(j))$ can only decrease the assignment cost. We use this in both subcases below.
    \begin{itemize}
        \item If $z(S_j^t) \geq 1$, set $x'_t(\sec_t(j),j) = 1-z(\prim_t(j))$, and remaining $x'_t(i,j)$'s to zero. The assignment cost becomes $d_t(\prim_t(j),j)z(\prim_t(j)) + d_t(\sec_t(j),j)(1-z(\prim_t(j)))$. Since $z(S_\ell^t) = 1$, we can write the second term as
        \begin{align*}
            &d_t(\sec_t(j),j)\cdot\paren{z(i_0) + z(\sec_t(j)) - z(\prim_t(j))}\\
            \leq &d_t(i_0,j)z(i_0) + d_t(\sec_t(j),j)z(\sec_t(j)) - d_t(\sec_t(j))z(\prim_t(j))\\
            = &\sum_{i \in S_\ell^t}d_t(i,j)z(i) - d_t(\sec_t(j))z(\prim_t(j))\,.
        \end{align*}
        Adding back the first term, the overall assignment cost is then at most $A_t(z,j)$.
        \item If $z(S_j^t) < 1$, set $x'_t(\sec_t(j),j) = z(\sec_t(j))$, $x'_t(i_0,j) = 1 - z(\prim_t(j)) - z(\sec_t(j))$, and remaining $x'_t(i,j)$'s to zero. Note that, since $z(S_\ell^t) = 1$, $x'_t(i_0,j) = z(i_0) - z(\prim_t(j)) \leq z(i_0)$, as required. We first note that
        \begin{align*}
            &d_t(i_0,j)\cdot\paren{z(i_0) - z(\prim_t(j)) - z(\sec_t(j))}\\
            \leq &d_t(i_0,j)z(i_0) - d_t(\sec_t(j),j)z(\prim_t(j)) - d_t(\sec_t(j),j)z(\sec_t(j))\,.
        \end{align*}
        Adding this to $d_t(\prim_t(j))z(\prim_t(j)) + d_t(\sec_t(j))z(\sec_t(j))$, the assignment cost becomes
        \begin{align*}
            d_t(i_0,j)z(i_0) + \paren{d_t(\prim_t(j),j) - d_t(\sec_t(j),j)}z(\prim_t(j)) \leq A_t(z,j)\,.
        \end{align*}
    \end{itemize}
\end{proof}

\begin{lemma}\label{lem:A-upper-bd}\label{lem:25}
    Given half-integral facility masses $\hat y \in \set{0,\half, 1}^F$, for each $t \in \set{1,2}$ and each $j \in \Reps_t$, $A_t(\hat y, j) \leq 2\hat C_t(j)$.
\end{lemma}
\begin{proof}
    Let $j$ be in $\nbr_t(\ell)$. Via \Cref{assn:completeness-again}, we have $\hat y \in \set{0,\half}^F$.
    
    In the case $\prim_t(j) \in S_\ell^t$, let $i_0$ be the other facility in $S_\ell^t$. Then,
    \begin{align*}
        A_t(\hat y, j) &= \half\paren{d_t(\prim_t(j),j) + d_t(i_0,j)} \leq \half\paren{2d_t(\prim_t(j),j) + d_t(\prim_t(j),\ell) + d_t(i_0,\ell)}\\
        &\leq d_t(\prim_t(j),j) + \hat C_t(\ell) \leq \half\paren{d_t(\prim_t(j),j) + d_t(\sec_t(j),j)} + \hat C_t(\ell)\\
        &= \hat C_t(j) + \hat C_t(\ell) \leq 2\hat C_t(j)\,.
    \end{align*}
    In the case $\prim_t(j) \notin S_\ell^t$, we know $\sec_t(j) \in S_\ell^t$. Once again, let $i_0$ be the other facility in $S_\ell^t$. Then,
    \begin{align*}
        A_t(\hat y, j) &= \half\paren{d_t(\sec_t(j),j) + d_t(i_0,j)} + \half\paren{d_t(\prim_t(j),j) - d_t(\sec_t(j),j)}\\
        &= \half\paren{d_t(i_0,j) + d_t(\prim_t(j),j)}\\
        &\leq \half\paren{d_t(i_0,\ell) + d_t(\ell, \sec_t(j)) + d_t(\sec_t(j),j) + d_t(\prim_t(j),j)}\\
        &= \hat C_t(\ell) + \hat C_t(j) \leq 2\hat C_t(j)\,.
    \end{align*}
\end{proof}

Now, we want an integral $\tilde y \in \set{0,1}^F$ s.t. for each $t \in \set{1,2}$,

$A_t(\tilde y) \leq O(1)\cdot \sum_{j \in \Reps_t}\abs{\child_t(j)}\hat C_t(j)$. We define the following polytope:

\[\calR := \set{z \in \RR_{\geq 0}^F : z(F) \leq k;\,\forall t \in \set{1,2}, \forall \ell \in \Super_t, z(S_\ell^t) = 1}\,.\]

Since the constraints in $\calR$ are over two laminar families, its constraint matrix is totally unimodular. The constant terms in the constraints are all integral, so we have via properties of total unimodularity~\cite{HoffmK1956} that
\begin{observation}\label{obs:R-integral}
    $\calR$ has integral extreme points.
\end{observation}
The remaining algorithm is analogous to the previous phase.
\begin{lemma}
	There is a polynomial time algorithm which returns an integral solution $(\tilde x_1, \tilde x_2, \tilde y)$ of \ref{kmed-cp:obj} s.t. $\forall t \in \set{1,2}$,	 $\sum_{j \in \Reps_t}\abs{\child_t(j)}\sum_{i \in F}d_t(i,j)\tilde x_t(i,j) \leq 24 \cdot \sum_{v \in C}C_t(v)$.
\end{lemma}
\begin{proof}
    Let $\hat y$ be the half-integral solution obtained via \Cref{lem:half-integral}. 
    
    Now for a $j \in \Reps_t$, by \Cref{clm:T-sample} and \Cref{obs:R-integral}, we can obtain an integral $\tilde y \in \calR$ such that $A_1(\tilde y) \leq 2A_1(\hat y)$ and $A_2(\tilde y) \leq 2A_2(\hat y)$. This is because, in \Cref{clm:T-sample}, we can plug in $A_1, A_2$ to be our linear functions, $\calR$ to be our polytope, $\hat y$ to be our starting point $p$, and $\tilde y$ to be the sampled extreme point.
    So by \Cref{lem:A-proxy,lem:A-upper-bd}, we can construct $\tilde x_t$'s s.t., for any $\eps' > 0$,
    \begin{align*}
        &\quad\sum_{j \in \Reps_t}\abs{\child_t(j)}\sum_{i \in F}d_t(i,j)\tilde x_t(i,j) \leq A_t(\tilde y) \leq 2A_t(\hat y)\\
        &\leq 2\cdot 2 \sum_{j \in C}\abs{\child_t(j)}\hat C_t(j) = 4 \sum_{j \in \Reps_t}\abs{\child_t(j)}\sum_{i \in F}d_t(i,j)\hat x_t(i,j)\\
        &\leq 4\cdot 6 \cdot \sum_{v \in C}C_t(v) \leq 24 \cdot \sum_{v \in C}C_t(v)
    \end{align*}

    where the penultimate step is from \Cref{lem:half-integral}.
\end{proof}

This, combined with \Cref{lem:filtering-4}, means that we obtain an integral solution

$S := \set{i \in F: \tilde y(i) = 1}$ such that $\cost_1(d_1; S)  + \cost(d_2;S) \leq 28 \cdot \lp$. This proves \Cref{thm:k-med:const-appx}.

\section{\texorpdfstring{FPT($k,T$) constant factor approximation}{FPT(k,T) constant factor approximation}}
\label{sec:fpt-3-appx}

\begin{theorem}\label{thm:3-appx}
    For any homogenous aggregator $\Psi$, any $z \in \NN \cup \set{\infty}$, and any constant error parameter $\e > 0$, 
    there is an algorithm that computes a $(3+\eps)$-approximate solution to 
    $\Psi$-aggregate $(k,z)$-clustering
    in time $ O(\eps^{-1}k)^{O(kT)} \cdot 2^{O(k^2T^2)}\cdot \poly(n)$.
\end{theorem}
\begin{proof}
We adapt an existing FPT $3$-approximation~\cite[Section 2.1]{cohen2019tight} for standard $(k,z)$-clustering (i.e. $T=1$).
Fix an error parameter $\eps > 0$. 

In our aggregate cluster instance, let $C$ denote the client set and let $F$ denote the facilities. For each scenario $t \in [T]$ (defined by a metric $d_t$ and a client weights $w_t$), we first compute a \emph{coreset} $C_t \subseteq C$
of size $O(\eps^{-2}k \log n)$ \cite{feldman2011unified}; that is, obtain a new client set $C_t$ and new weights $w'_t$ on $C_t$ with the property that for any subset of facilities $S \subseteq F$,
\[\sum_{u \in C_t}w'_t(u)\cdot d_t(u,S) \in (1\pm \eps) \cdot \sum_{v \in C}w_t(v)\cdot d_t(v,S)\,.\]
Henceforth we work with the coresets $C_t$ rather than the entire client set $C$ in each metric.

Let $\OPT \subseteq F$ be the optimal solution to the aggregate clustering instance, with (unknown) centers $o_1,\dots,o_k$. For an $i \in [k]$ and $t \in [T]$, let $C_{i,t}$ be the cluster served by $o_i$ under $d_t$. Define the client $\ell_{i,t}$ to be the closest client (in the coreset $C_t$) to $o_i$, that is $d_t(o_i,C_{i,t}) = d_t(o_i,\ell_{i,t})$, and define $r_{i,t} \coloneqq d_t(o_i,\ell_{i,t})$.

For each scenario $t \in [T]$ and index $i \in [k]$, we can guess (by enumeration) the client $\ell_{i,t}$ in time $\genfrac(){0pt}{1}{O(\eps^{-2}k \log n)}k$.
We can't guess $r_{i,t}$ exactly, but we can guess it up to the closest power of $(1+\e)$: we assume WLOG (up to a $1+\e$ loss in the approximation factor) that the aspect ratio of $d_t$ is $\poly(n)$, and so bucketing pairwise distances of $d_t$ by powers of $(1+\eps)$ reduces the number of possible distances to $O(\log_{1+\e} n) = O(\e^{-1}\log n)$. Overall, in time $(\genfrac(){0pt}{1}{O(\eps^{-2}k \log n)}k\cdot O(\eps^{-1}\log n)^k)^T$, we can guess each $\ell_{i,t}$ and $r_{i,t}$. By \Cref{fact:loglog-trick}, this overhead is $O(\eps^{-1}k)^{O(kT)} \cdot 2^{O(k^2T^2)}\cdot \poly(n)$ $=O_{\eps,k,T}(1)\cdot \poly(n)$.

Suppose our guesses are correct. For each $i \in [k]$, consider the set $F_i$ containing the facilities that are within distance $r_{t,t}$ of $\ell_{i,t}$ in every metric $d_t$, i.e. $F_i := \cap_{t=1}^T B_t(\ell_{i,t}, r_{i,t}) \cap F$. If our guesses are correct, then $F_i$ is non-empty since $o_i$ lies in this intersection. Let our solution $\ALG$ consist of an arbitrary $f_i \in F_i$ for each $i \in [k]$. For a client $v$ in the cluster $C_{i,t}$, we have $d_t(o_i, \ell_{i,t}) \leq d_t(v,o_i)$ by definition of $\ell_{i,t}$; and also $d_t(\ell_{i,t},f_i) \leq (1+O(\e)) \cdot r_{i,t} = (1+O(\e)) \cdot d_t(\ell_{i,t},o_i)$. Using these, we get
\begin{align*}
    d_t(v,\ALG) & \leq d_t(v,f_i) \leq d_t(v,o_i) + d_t(o_i, \ell_{i,t}) + d_t(\ell_{i,t}, f_i)\\
    & \leq d_t(v,o_i) + (2 + O(\e))\cdot d_t(o_i, \ell_{i,t}) \leq (3+ O(\e))\cdot d_t(v,o_i)
\end{align*}
for every $v \in C$ and $t \in [T]$. Since the $(k,z)$-clustering cost and the aggregator $\Psi$ are both homogeneous functions, this yields a $(3 + O(\e))$-approximation.

Rescaling $\e$ by a constant yields the desired result.
\end{proof}

\subsection*{Simpler algorithm in a special case}

In the special case of $k$-supplier, we have a simpler $(3+\eps)$-algorithm that doesn't utilize coresets, and runs in $O(\eps^{-O(T^2)})\cdot O(k)^{O(kT)} \cdot \poly(n)$ time. This algorithm does not require coresets for its enumeration component.

\textit{Guessing radii.} We use the same bucketing as above, which contributes a $(1+\eps)$ factor and allows us to assume that there are $O\paren{\frac{\log n}\eps}$ distances per metric.

Now let $\opt_t := \cost(d_t;S)$ for each $t \in [T]$, and note that each $\opt_t$ is some $d_t(u,v)$ times either $w_t(v)$ or $w_t(u)$. So with $O\paren{\frac{\log n}\eps}^T$ overhead, which due to \Cref{fact:loglog-trick} is an $\eps^{-O(T^2)}\cdot \poly(n)$ overhead, we can guess $\opt_t$ in each metric. Then for each $v \in V$ and $t \in [T]$, let $r_t(v) := \opt_t/w_t(v)$ be the radius within which, under $d_t$, the client $v$ must be served.

\textit{Guessing clusters.} For each $t \in [T]$, we run Plesn\'ik's~\cite{Ples1987} algorithm in each metric $d_t$ with radii $\set{r_t(v)}_{v \in V}$. These Plesn\'ik subroutines yield, for each $t \in [T]$:
\begin{itemize}
    \item Clusters $C_{1,t},\dots,C_{k'_t,t} \subseteq C$ for some $k'_t \leq k$;
    \item A center $h_{i,t} \in C_{i,t}$ for each cluster, such that $C_{i,t} \subseteq B_t(h_{i,t}, 2r_t(h_{i,t}))$ and $\forall v \in C_{i,t}$, $r_t(v) \geq r_t(h_{i,t})$;
    \item For any two distinct clusters $C_{i,t}$ and $C_{j,t}$, $d_t(h_{i,t},h_{j,t}) > r_t(h_{i,t}) + r_t(h_{j,t})$.
\end{itemize}

Thereafter, for each $t \in [T]$, we augment the clusters above with empty clusters $C_{k'_t+1,t},\dots,C_{k+1}$; and call the resulting partition $\calC_t$. Note that each $\calC_t$ has at most $k$ non-empty clusters.

Suppose the (unknown) centers in $\OPT$ are $o_1,\dots, o_k$. With $(k+1)!^T$ overhead, we can guess one cluster per metric whose center is served by $o_i$; and rename WLOG so that this cluster is called $C_{i,t}$. That is, for each $i \in [k]$ and $t \in [T]$, we guess $C_{i,t} \in \calC_t$ such that $d_t(h_{i,t},\OPT) = d_t(h_{i,t},o_i)$. This is possible to guess in the desired runtime because any nonempty $C_{i,t}$ is unique in $\calC_t$; to see this, suppose not, and consider the distinct 
centers $h_{i,t}, h_{j,t}$ served by $o_i$ under $d_t$. Then $d_t(h_{i,t}, h_{j,t}) \leq d_t(h_{i,t},o_i) + d_t(h_{j,t}, o_i) \leq r_t(h_{i,t}) + r_t(h_{j,t})$, violating one of the Plesn\'ik properties listed above.

\textit{Picking centers.} For each $i \in [k]$, let $F_i = F \cap \bigcap_{t \in [T]: C_{i,t} \neq \emptyset}B_t(h_{i,t}, r_t(h_{i,t}))$. Note that for each $t \in [T]$ and $v \in V$, we have $r_t(v) \geq \cost(d_t;\OPT) / w_t(v)$. In particular, this holds for each $h_{i,t}$, showing that each $F_i$ is non-empty. So let our answer $\ALG$ consist of an arbitrary $c_i$ from each $F_i$.

\textbf{\textbf{$3$-approximation.}} Fix a $t \in [T]$ and a $v \in C_{i,t}$. Then by properties of Plesn\'ik's algorithm listed above, 
and our construction of $\ALG$, $d_t(v,\ALG) \leq d_t(v,c_i) \leq d_t(v,h_{i,t}) + d_t(h_{i,t},c_i) \leq 3r_t(v) \implies w_t(v)\cdot d_t(v,c_i) \leq 3\opt_t$. The homogeneity of $\Psi$ yields the desired approximation.

The algorithm runs in time $O(\eps^{-O(T^2)})\cdot O(k)^{O(kT)} \cdot \poly(n) = O_{\eps,k,T}(1)\cdot\poly(n)$.

\section{FPT Approximations in Well-structured Metrics}
\label{sec:structured}

\subsection{Metrics with Bounded Scatter Dimension}\label{sec:epas}

In this section, we study aggregate clustering in the case where all the $T$ metrics have bounded \emph{$\e$-scatter dimension}. This notion of metric dimension was originally introduced by \cite{ABB+23} to obtain $O_{k,\e}(1)\cdot\poly(n)$-time $(1+\eps)$-approximation algorithms (ie, efficient parameterized approximation schemes, or EPASes for short) for a large class of center-based $k$-clustering problems, for example the $T=1$ case of $(k,z)$-clustering.
\begin{definition}[$\e$-scattering, $\e$-scatter dimension{~\cite[Definition IV.1]{ABB+23}}]
    An \emph{$\e$-scattering} in a metric space $M = (C \cup F, d)$ is a sequence $(x_1,p_1),(x_2,p_2),\dots,(x_l,p_l)$ s.t.
    \begin{itemize}
        \item $\forall i \in [l]$, $x_i \in F$ and $p_i \in C$;
        \item \textnormal{[covering]} $\forall 1 \leq j < i \leq l$, $d(x_i,p_j) \leq 1$; and
        \item \textnormal{[$\e$-refutation]} $\forall i \in [l]$, $d(x_i,p_i) > 1+\e$.
    \end{itemize}
    The $\e$-scatter dimension of a family of metric spaces $\calM$ is the maximum length of any $\e$-scatter of an $M \in \calM$.
\end{definition}

A key component of the ($T=1$) EPASes of \cite{ABB+23} is an EPAS for the special case of unweighted $k$-supplier. We recall their algorithm, following their presentation. At a high level, they maintain a $k$-sized solution $X \subseteq F$, and while this solution is not a $(1+\e)$-approximation, clients with high cost are identified to guide a recomputation of $X$; the bounded $\e$-scatter dimension is used to show a bound on the number of recomputations needed.
In more detail, $X$ is initialized arbitrarily as $X = \set{x_1,\dots,x_k}$. Empty buckets $B_1,\dots,B_k$ are also initialized, corresponding to clusters that would eventually be served by $x_1,\dots,x_k$ respectively. We assume we know the cost $\opt$ of the optimal solution\footnote{We can approximate it by bucketing by powers of $1+\e$ and guessing over the $O(\e^{-1}\log n)$ options.}.
While there exists a client violating the desired $(1+\e)$-approximation of $\opt$---i.e. while $\exists v \in C : d(v,X) > (1+\e)\opt$---the algorithm guesses the index $i \in [k]$ of the optimal cluster that serves $v$ in $\OPT$. The client $v$ is then placed in the bucket $B_i$, and $x_i$ is recomputed so that $\forall u \in B_i, d(u,x_i) \leq \opt$. If such an $x_i$ cannot be obtained, then the algorithm asserts failure and restarts from the initialization; but with positive probability $1/k$ we don't fail (because with probability $1/k$ we guess $i$ correctly).
Now here is the key insight, that each bucket $B_i$ can't grow too large. Indeed, consider the sequence of requests $(v_i^{(1)}, \ldots, v_k^{(\ell)})$ added to $B_i$ over time, and consider the sequence of candidate locations for $x_i$ over time, $(x_i^{(1)}, \ldots, x_i^{(\ell)})$. These two sequences form an $\e$-scatter (after we normalize distances so $\opt = 1$). In particular, if the $\e$-scatter dimension of $d$ is upper-bounded by some $\lambda(\e)$, then every bucket $B_i$ only grows to size $\lambda(\e)$ until we either assert success or failure. As there there are $k$ buckets overall, each iteration of the algorithm requires at most $k\lambda(\e)$ random choices each of which is correct with probability $1/k$; so with constant probability, we successfully find a solution after $O_{\e, k}(1)$ re-initializations.

For our aggregate clustering problem, we generalize the above. Rather than guessing a single partition of the clients into clusters served by $x_1, \ldots, x_k$, we need to guess $T$ different partitions, one for each metric: we now have buckets $B_{i,t}$ for each $i \in [k], t \in [T]$, and we separately guess $\opt_t = \cost(d_t;\OPT)$ for each $t$.
These metric-wise buckets allows us to mostly decouple the different metrics: when we find a violation in metric $d_t$, i.e. get $\cost(d_t;X) > (1+\Theta(\e))\opt_t$, we select a random index $i\in[k]$ and add to bucket $B_{i,t}$, just as is done for the $T=1$ case. Only when recomputing $x_i$ do we take the other metrics into account; we choose $x_i$ to obey the constraints of all the buckets $B_{i,t}$ across all $t\in[T]$, that is $\forall t\in[T] \forall u \in B_{i,t}, \dist_t(u,x_i) \le \opt_t$.
Observe that each bucket $B_{i,t}$ still has bounded size, as it provides witness of an $\e$-scatter in metric $d_t$. To be precise, if one considers the sequence of requests $(v_i^{(1)}, \ldots, v_i^{(\ell}))$ in $B_{i,t}$ and the sequence of candidate positions for $x_i$ \emph{at the moment immediately before request $v_i^{j}$ was added}\footnote{We remark that, unlike the $T=1$ case, the center $x_i$ could take on more than $\ell$ candidate positions, because the the position of $x_i$ could be changed due to a request in any bucket $B_{i,t'}$ for $t' \in [T]$; nevertheless, even if $x_i$ is changed due to a different bucket $B_{i,t'}$ we still guarantee that $x_i$ covers requests in $B_{i,t}$, and so one can construct an $\e$-scatter.}, $(x_i^{(1)}, \ldots, x_i^{(\ell)})$, then these sequences form an $\e$-scatter for $d_t$.
In particular, each iteration of the algorithm makes $kT\lambda(\e)$ random choices, and so it succeeds with constant probability after $O_{\e, k, T}(1)$ rounds.

\begin{theorem}
\label{thm:epas-basic}
    When all $T$ metrics have $\e$-scatter dimension at most $\lambda(\e)$, for any homogeneous aggregator $\Psi$ there is an algorithm for $\Psi$-aggregate (unweighted) $k$-supplier runs in time $O_{\eps,k,T}(1)\cdot\poly(n)$, and produces a $(1+\e)$-approximation with constant probability.
\end{theorem}

When $T=1$, the $k$-supplier EPAS generalizes~\cite{ABB+23} to weighted $k$-center (by initializing with the help of a $O(1)$-approximation), then to weighted $k$-median (by a random sampling argument), and then to more general $k$-clustering problems. We adapt some of these ideas to the aggregate setting, and prove the following theorem.

\begin{restatable}{theorem}{epas}
\label{thm:epas}
    When all $T$ metrics have $\e$-scatter dimensions upper-bounded by some $\lambda(\e)$, for any homogeneous aggregator $\Psi$ there is an algorithm for $\Psi$-aggregate $k$-median that runs in time $O_{\eps,k,T}(1)\cdot\poly(n)$, and produces $(1+\e)$-approximations with constant probability.
\end{restatable}

Our algorithm for proving \Cref{thm:epas} is a direct extension of the analogous algorithm at $T=1$~\cite{ABB+23}, so we recall the $T=1$ unweighted $k$-supplier algorithm that we summarized from the literature~\cite{ABB+23} in \Cref{sec:epas}. Let us understand how this is extended~\cite{ABB+23} to weighted $k$-supplier and weighted $k$-median at $T=1$.

For the $k$-supplier problem with weights, $X$ is initialized as Plesn\'ik's $3$-approximation~\cite{Ples1987}, and labeled as $X = \set{c_1,\dots,c_k}$. As in the unweighted case, empty buckets $B_1,\dots,B_k$ are also initialized, corresponding to clusters that would eventually be served by $c_1,\dots,c_k$ respectively.
While there exists a client violating the desired $(1+\e)$-approximation---i.e. while $\exists v \in C : d(v,X) > (1+\e)\opt/w(v)$, where $\opt$ is a guess of the optimum---the algorithm guesses the index $i \in [k]$ of the optimal cluster that serves $v$ in $\OPT$. The client-radius pair $(v,\opt/w(v))$ is then placed in the bucket $B_i$, and $c_i$ is recomputed so that $\forall u \in B_i, d(u,c_i) \leq \opt/w(u)$. If such a $c_i$ cannot be obtained, then the algorithm asserts failure and restarts with a fresh $3$-approximation; but this failure probability is $(1-1/k)$, i.e. the probability that $i$ was guessed incorrectly. By construction of $B_i$, pairs in it with the same radius (say $r$) appear in an $\e$-scatter (after scaling distances down by $r$); and it can be shown~\cite{ABB+23} that there are only $O_\e(1/\e\cdot\log 1/\e)$ different radii appearing in $B_i$ because we initialize $X$ as a 3-approximation. So if the $\e$-scatter dimension is upper-bounded by some $\lambda(\e)$, then iterations involving $B_i$ succeed with probability $(1/k)^{O(1/\e\cdot\log 1/\e)\cdot\lambda(\e)}$, leading to an overall success probability of $(1/k)^{O(1/\e\cdot\log 1/\e)\cdot k\lambda(\e)}$ each time we reinitialize.

The above approach for $k$-supplier is then generalized to other $k$-clustering problems, including $k$-median, by computing an upper bound $u(v)$ for each $d(v,\OPT)/w(v)$, and running the above algorithm with these $u(v)$'s in place of $\opt/w(v)$. The crucial modification is that violation to the desired $(1+\e)$-approximation is no longer clientwise---as $\opt$ is no longer a single client's cost, a violation may take a form like $\sum_{v \in C}d(v,X) > (1+\e)\opt$. Nevertheless, once a cost violation is found, a clever sampling~\cite[Algorithm 1]{ABB+23} is possible to pick a client $v$ that is indeed a violation on its own, i.e. $d(v,X) > (1+\e')d(v,\OPT)$ for some $\e' = \Theta(\e)$. Then $(v,(1+\e')d(v,X))$ is added to bucket $B_i$ for a guessed $i \in [k]$. The sampling also ensures that $u(v) \geq d(v,\OPT)$ is small for a sampled $v$, so that smaller radii are more likely to be added to a bucket. Once a pair $(v,r)$ is added to a bucket, any subsequent $(v,r')$ in the same bucket must satisfy $r' < r$; so ensuring smaller radii upper-bounds the length of a bucket, as desired.

We now extend the above to $T > 1$, proving \Cref{thm:epas} (restated below). We keep our analysis brief, focusing on the differences from the $T=1$ case~\cite{ABB+23}.

\begin{proof}[Proof of \Cref{thm:epas}]

As before, we assume via bucketing that there are $O(\e^{-1}\log n)$ distances per metric, allowing us to guess $\opt_t = \cost(d_t;\OPT)$ for each $t \in [T]$ within the desired running time. We also compute upper bounds $u_t(v)$ for every $d_t(v,\OPT)/w_t(v)$ as follows: if $z = \infty$ then $u_t(v) = \opt_t/w_t(v)$; otherwise, $u_t(v) = 2 \cdot \min\set{r : w_t(B_t(v,r)) \geq 3\opt_t/r}$. Using these $u_t(v)$'s as our radii, we run our polytime $3$-approximation algorithm from \Cref{sec:T=2:ksupp}; that is, we obtain $X = \set{c_1,\dots,c_k} : \forall t \in [T], v \in C, d_t(v,X) \leq 3u_t(v)$.

Observe that, by homogeneity of $\Psi$, if we had $\cost(d_t;X) \leq (1+\e)\opt_t$ for every $t \in [T]$, then we would be done. So there must be some $t$ violating this inequality, i.e. such that $\cost(d_t;X) > (1+\e)\opt_t$; and we can identify this $t$ in $\poly(n,T)$ time by simply computing all $T$ costs. While such a $t$ exists, we run a \emph{violation loop} to improve the solution $X$. An iteration of this loop begins with an attempt to find a $v \in C$ with the following properties w.r.t. $d_t$ and $X$:
\begin{itemize}
    \item (witness) $d_t(v,X) > (1+\e/3)d_t(v,\OPT)$, and
    \item (upper bound) $u_t(v) = O(kd_t(v,X)/\e)$.
\end{itemize}

For $k$-supplier, this $v$ is easy to find, as it was in the $T=1$ case: $\cost(d_t;X)$ must equal some $w_t(v)\cdot d_t(v,X)$, so we pick that $v$. Then $d_t(v,X) = \cost(d_t;X) > (1+\e)\opt_t \geq (1+\e)d_t(v,\OPT)$, satisfying the witness condition. Also $u_t(v)\cdot w_t(v) = \opt_t \leq \cost_t(d_t;X) / (1+\e) = w_t(v)\cdot d_t(v,X) / (1+\e)$, implying $u_t(v) \leq d_t(v,X) / \e$, which satisfies the upper bound condition.

When $z \in \NN$, we observe the following, again in line with the $T=1$ case: to satisfy the witness condition alone, sampling with probability $\propto w_t(v)\cdot d_t(v,X)$ from $C$ suffices via an averaging argument. But to satisfy both conditions, we instead sample from a smaller client-set $A_t$ in which the upper bound condition already holds. The key technical result, like in the $T=1$ case, is that ``many witnesses lie in $A_t$'' or, more precisely, that a large fraction of the weighted cost of $A_t$ is contributed by witnesses. This result involves the exact same construction and analysis as the analogous result~\cite[Lemma V.10]{ABB+23} in the $T=1$ case.
Where that proof constructed desirable ``heavy witness sets'' $W^+_1,\dots,W^+_k \subseteq A$, we now have for each $t \in [T]$, analogous sets $W^+_{1,t},\dots,W^+_{k,t} \subseteq A_t$, satisfying the same guarantees within that $d_t$. These constructions in the $T$ different metrics do not interact with one another, so we can analyze each one exactly as for the $T=1$ analysis. Because of this direct similarity, we omit the details of this proof, and only state the extended result below.
\begin{lemma}[Generalization of{~\cite[Lemma 5.10]{ABB+23}} to $T > 1$]\label{lem:epas}
    Fix a violating metric, i.e. a $t \in [T] : \cost(d_t;X) > (1+\e)\opt_t$. Let $v$ be a client sampled from the set $A_t := \set{v \in C : d_t(v,X) = \Omega\paren{\frac{\e\cdot u_t(v)} k}}$,
    with probability proportional to $w_t(v)\cdot d_t(v,X)$. Then with probability $\Omega(\e)$, $v$ satisfies the witness condition and the upper bound condition.
\end{lemma}

Having sampled this $v$ with the desired properties, we guess its ``correct cluster'' $i \in [k]$, i.e. where the unknown optimal centers are $o_1,\dots,o_k$, we guess that $d_t(v,\OPT) = d_t(v,o_i)$. We then add $(v,r_t(v))$ to the bucket $B_{i,t}$; here $r_t(v) = u_t(v)$ for $k$-supplier, and $r_t(v) = d_t(v,X)/(1+\e/3)$ otherwise. So $r_t(v) < d_t(v,X)$. Then, we attempt to recompute $c_i$ such that \[\forall t' \in [T], p \in B_{i,t'}, d_{t'}(p,c_i) \leq r_{t'}(p).\] This is the \emph{only major step where multiple metrics are involved}; it still holds, however, that if each iteration guesses the correct $i$ then this $c_i$ exists. If such a $c_i$ does not exist, then we assert a \textsc{FAIL} state, and restart with a fresh initialization of $X$; otherwise we continue looping by looking for another violation, which might occur at a different $t' \in [T]$.

To establish correctness, we must bound the failure probability of this algorithm: note that each guess of $i \in [k]$ succeeds with probability $(1/k)$, so for iterations that augment a particular bucket $B_{i,t}$, we succeed with probability $(1/k)^{\abs{B_{i,t}}}$. To bound this bucket size, observe that for every radius $r$ that appears in $B_{i,t}$, the set $B_{i,t}\restrict r := \set{(v,r_v) : r_v = r}$ is an $\e$-scatter in $d_t$ (when scaled down by $r$). That is, $\abs{B_{i,t}\restrict r} \leq \lambda(\e)$. So we must bound the number of different radii appearing in $B_{i,t}$. Since the bucket has a fixed $t$, this follows exactly as for the $T=1$ case~\cite[Section II, p. 1384]{ABB+23}, yielding $O(1/\e\cdot\log 1/\e)$ different radii in $B_{i,t}$. So over the $kT$ buckets, we obtain a success probability of $(1/k)^{O_\e(1)\cdot kT}$, meaning that we have succeeded with constant probability after $O_{k,\e,T}(1)$ repetitions.
\end{proof}

\begin{remark}[Barrier at algorithmic scatter dimension]
The authors of \cite{ABB+23} also consider a weaker notion of \emph{algorithmic $\e$-scatter dimension} (where, roughly speaking, the length of $\e$-scatters are only bounded if the sequence of centers $x_i$ is chosen according to a certain algorithm), and their EPASes extend to metrics of bounded algorithmic scatter dimension. For example in the case of $k$-supplier, when recomputing the center $c_i$ for some bucket $B_i$, they just use the algorithm for computing the center rather than choosing $x_i$ to be an arbitrary point that covers $B_i$.
We are aware of two follow-up works \cite{GI25, BGI25} that build on the framework of \cite{ABB+23}, both of which apply also for algorithmic scatter dimension. In contrast, our aggregate algorithm \emph{does not extend} to the more general setting; this is because we recompute $x_i$ in a way that covers all the buckets $B_{i,t}$ across $t \in [T]$ (and such an $x_i$ is guaranteed to exist, if we guessed the buckets correctly), whereas an algorithmic version would require that $x_i$ was chosen to simultaneously satisfy the algorithms of all $T$ metrics (and such an $x_i$ does \emph{not} necessarily exist).
\end{remark}

\subsection{Bounded-treewidth Graphs}

\begin{restatable}{theorem}{treewidth}
\label{thm:treewidth}
    When the base graph $G$ has treewidth $\tw$ and the aggregator $\Psi = $ sum, then:
    \begin{itemize}
        \item there is an exact algorithm for sum-aggregate $(k,z)$-clustering that runs in time $n^{O(T \cdot \tw)}$.
        \item for any $\eps > 0$, there is a $(1+\eps)$-approximation algorithm for sum-aggregate $(k,z)$-clustering that runs in time $(\frac{\log n}{\eps})^{O(T \cdot \tw)} \cdot \poly(n) = O_{\eps, T, \tw}(1) \cdot \poly(n)$.
    \end{itemize}
\end{restatable}
In this section we focus on the $k$-median case; the ideas extend naturally to $(k,z)$-clustering and $k$-supplier. We prove \Cref{thm:treewidth} by adapting a folklore dynamic program for solving $k$-median on bounded-treewidth graphs. (Recall that a tree decomposition of a graph $G$ is a tree $\mathcal{T}$ of \emph{bags}, where each bag is a subset of vertices $V(G)$, that satisfies two properties: (1) for every edge $(u,v)$ in $G$, there is a bag containing both $u$ and $v$, and (2) for every vertex $v$ in $V(G)$, the bags containing $v$ induce a connected subtree of $\mathcal{T}$. The \emph{width} of the tree decomposition is 1 less than the maximum bag size, and the \emph{treewidth} of a graph is the minimum width of any tree decoposition. A tree has treewidth 1.) To the best of our knowledge, the DP for $k$-median has not been written down explicitly in the setting of bounded treewidth graphs, but the technique is standard. Our presentation is based on a DP that appeared~\cite{ARR98,KR07,CFS21} as part of a PTAS for $k$-median in Euclidean metrics. In this section we explain the main idea of the algorithm and defer the formal details to Appendix~\ref{sec:treewidth-full}.

We begin by reviewing the folklore DP. Every bag $B$ of the tree decomposition $\calT$ acts as a small separator between an ``inside'' piece (induced by the vertices in the subtree of $\calT$ rooted at $B$) and an ``outside'' piece (induced by all vertices outside the subtree rooted at $B$). For each separator, the DP guesses the \emph{interface} of the optimal solution $\OPT$ between the two pieces. In the case of $k$-median, the interface of bag $B = \set{v_1, \ldots, v_{\tw}}$ consists of the distance of every $v_i$ to the closest open center in the inside (resp. outside) piece, as well as the number of centers opened in the inside (resp. outside) piece in the optimal solution.
Note that there are $n^{O(\tw)}$ possible interfaces for any bag, as there are $O(n^2)$ possible distances for each $v_i$; by rounding distances to the nearest power of $(1 + \frac{\eps}{\log n})$, we could reduce the number of interfaces to $(\frac{\log n}{\eps})^{O(\tw)}$. 

The DP table contains a cell for every bag $B$ and every possible interface of $B$, where the value of the cell represents the cost paid by vertices in the inside piece under the best solution that respects the interface.
The value of a DP cell at bag $B$ can be computed from the DP cells of the children of $B$, by taking a minimum over all combinations of children interfaces which are \emph{consistent} with the interface of $B$ (checking the consistency of interfaces can be done in a straightforward manner in polynomial time \cite{ARR98, CFS21}). The overall runtime is $n^{O(\tw)}$,
or $(\frac{\log n}{\eps})^{O(\tw)} \cdot \poly(n)$ if we tolerate $1+\eps$ approximation.

To adapt the DP to the aggregate setting, we simply change the notion of \emph{interface}: rather than guessing just one interface, we guess a different interface for each of the $T$ metrics $d_1, \ldots, d_T$.
The cells of our DP table now consist of a bag $B = \set{v_1, \ldots, v_{\tw}}$ and a set of $T$ tuples; the $t$-th tuple consists of the distance \emph{with respect to $d_t$} of every $v_i$ to the closest (with respect to $d_t$) open center in $\OPT$ in each piece, as well as the cost paid by each piece in the $d_t$ metric under $\OPT$.
(Note: the reason that we need to guess a different interface for each metric is that, even though the optimal set of centers $\OPT$ for the aggregate problem stays consistent throughout the $T$ metrics, a vertex $v$ could be served by a different center of $\OPT$ in each different metric $d_t$.)
There are $n^{O(T \cdot \tw)}$ possible interfaces. The DP can be computed in a bottom-up fashion, as in the non-aggregate setting, for a total runtime of $n^{O(T \cdot \tw)}$.
If one allows $1+\eps$ approximation, the runtime improves to $(\frac{\log n}{\eps})^{O(T \cdot \tw)} \cdot \poly(n)$, which is at most $\eps^{-O(T^2 \cdot \tw^2)} \cdot \poly(n) = O_{\eps,T,\tw}(1) \cdot \poly(n)$ (see \Cref{fact:loglog-trick}).

\begin{remark}
\label{remark:treewidth-algo-limit}
    Observe that if $\tw = o\left(\frac{\log n}{\log \log n}\right)$, then \Cref{thm:treewidth} provides a $(1+\e)$-approximation algorithm in $2^{o(\log n) \cdot T \cdot \log(1/\e)}\cdot \poly(n) = O_{\e,T}(1) \cdot \poly(n)$ time (by \Cref{fact:loglog-trick}).
\end{remark}

\input{treewidth}



\bibliographystyle{alpha}
\bibliography{references}

\appendix

\section{Useful facts}

\begin{fact}\label{fact:loglog-trick}
    For any $f(n) = n^{o(1)}$ and any variable $\tau \ge 1$, the function $(f(n))^\tau$ is bounded above by $O_{\tau}(1) + n$. 
    For example, when $f(n) = \log n$, we have $(\log n)^\tau \le 2^{\tau^2} + n$.
\end{fact}
\begin{proof}
    Let $g(n) = \log f(n)$, and note $g(n) = o(\log n)$. 
    If $\tau$ satisfies $k \tau \le \log n$, then clearly $f(n)^\tau \le n$. On the other hand, if $\tau > \frac{\log n}{k}$, then the fact that $g(n) = o(\log n)$ implies that $\tau$ grows as a function of $n$, and so $f(n)^\tau = O_{\tau}(1)$.

    We now give an example that provides explicit bounds for the case $f(n) = \log n$.
    If $\tau \le \log \log n$, then $(\log n)^\tau \le 2^{(\log \log n)^2} \le n$. Otherwise, if $\tau > \log \log n$, then $(\log n)^\tau = 2^{\tau \log \log n} \le 2^{\tau^2}$.
\end{proof}

\begin{observation}
\label{obs:generalized}
    Fix an aggregator function $\Psi$. Suppose that one can find an $\alpha$-approximation to the $\Psi$-aggregate $k$-supplier (or $k$-median, or $(k,z)$-clustering) problem on $T$ metrics and $n$ vertices in time $f(n, T, k, \alpha)$. Then one can find an $\alpha$-approximation to generalized aggregate clustering, with the restriction that weights come from $\set{0,1}$, in time $f(2n+1, T, k+1, \alpha) + \poly(n, T)$.
\end{observation}
\begin{proof}
    Suppose we are given an instance $\eta$ of $\set{0,1}$-weighted generalized aggregate $k$-clustering, defined by metrics $d_1, \ldots, d_T$, client set $C$, and facility set $F$. We construct an instance $\eta'$ of vanilla aggregate $(k+1)$-clustering as follows. Create a new vertex $r$, and let the new facility set $F'$ be $F' \coloneqq F \cup \set{r}$. 
    The client set $C'$  is set to be $C$. We assume WLOG that $C'$ and $F'$ are disjoint: if some vertex $v$ is both a client and facility, we duplicate $v$, and make one copy a client and one a facility. After this duplication, we have $|F' \cup C'| \le 2 |F\cup C| + 1 = 2n+1$.
    
    For every $t \in [T]$, we construct a new distance metric $d'_t$ as follows. We define the \emph{active vertices} to be the set of clients with weight 1 in the $t$-th instance, together with the set of facilities excluding $r$ (that is, $F' \setminus \set{r} = F$). The \emph{inactive vertices} are those clients with weight 0 in the $t$-th instance, and the facility $r$. For any two vertices $v_1, v_2 \in C \cup F$, we define
    \[\hat d'_t(v_1, v_2) \coloneqq  \begin{cases}
        \dist_t(v_1, v_2) &\text{if $v_1$ and $v_2$ are both active}\\
        \infty &\text{if one of $\set{v_1, v_2}$ is active and the other is inactive}\\
        0 &\text{if $v_1$ and $v_2$ are both inactive}
    \end{cases}\]

    One can compute a solution to the instance $\eta$ from a solution to $\eta'$. Indeed, suppose $S' \subseteq F'$ is a set of $k+1$ facilities with aggregate cost $\kappa$ in $\eta'$. It is straightforward to see that, if $\kappa'$ is finite, then facility $r$ is included in $S'$, and the set $S \coloneqq S' \setminus \set{r}$ has aggregate cost $\kappa$ for the instance $\eta$. Moreover, if $S \subseteq F$ is any set of $k$ facilities with aggregate cost $\kappa$ in $\eta$, then the set of facilities $S' \coloneqq S \cup \set{r}$ has aggregate cost $\kappa$ in $\eta'$. This proves the claim: in polynomial time one can construct $\eta'$, and in time $f(2n+1, T, k+1, \alpha)$ one can compute an approximate solution $S'$ and return $S' \setminus \set{r}$.
\end{proof}

\end{document}

%% file: treewidth.tex
\subsubsection*{Proof of \Cref{thm:treewidth}}\label{sec:treewidth-full}



We describe in detail the exact $n^{O(T \cdot \tw)}$ algorithm for stochastic aggregate $k$-median. At the end, we sketch the changes needed for the $1+\eps$ approximation, and the changes needed to extend the result to general $\Psi$-aggregate $(k,z)$-clustering and $k$-center as claimed in the theorem. Let $F \subseteq V(G)$ denote the facilities, and for every $t \in [T]$ let $w_t:V(G) \to \mathbb{R}_{>0}$ denote the weights of the clients in the $t$-th metric. We let $C \subseteq V(G)$ denote the set of \emph{non-trivial clients}, i.e. those clients $c$ with $w_t(v) > 0$ for at least one metric $t$.

Let $\calT$ be a rooted tree decomposition for $G$, where $\cT$ has width $O(\tw)$, depth $O(\log n)$, and is a binary tree. Such a $\calT$ can be found in time $2^{O(\tw)}\cdot\poly(n)$ \cite{BH98}.
Moreover, we assume without loss of generality (as in \cite{adamczyk2019constant}) that facilities $F$ and the non-trivial clients $C$ appear \emph{only} in leaf bags of $\cT$, and that each facility and non-trivial client appears in exactly one bag. This assumption can be enforced WLOG by \emph{duplicating} each client or facility $v$ to make a copy $v'$ connected to $v$ by $0$-weight edge, and updating the tree decomposition to contain a leaf bag $\set{v, v'}$; the vertex $v$ is treated as a Steiner vertex (ie., $v$ is treated if it was not a facility or client, meaning $v \not \in F$ and $w_t(v) = 0$ for all $t$), and the vertex $v'$ (which appears in exactly one bag in $\cT$, in a leaf bag) is treated as the client/facility. This duplication process could increase the arity of the tree decomposition $\cT$, but further duplication can reduce the arity back to 2 while only increasing the number of vertices in $G$ (and the depth of $\cT$) by a constant factor.

For each $t \in [T]$, let $R_t$ be the set of possible pairwise distances under the metric $d_t$. Note that $|R_t| = O(n^2)$.
For any node $b$ in $\calT$, let $X_b \subseteq V$ be the bag of vertices associated with the node. Let $\calT_b$ be the subtree rooted at $b$,
and let $V_b \subseteq V$ be the vertices appearing in bags of $\calT_b$.


\medskip \noindent \textbf{Definition of DP.}
Our DP entries are defined over the following:
\begin{itemize}
    \item A node $b$ in the tree decomposition $\calT$.
    \item A budget $k' \in \mathbb{N}$, representing the number of centers we can open in $\cT_b$.
    \item An ``outer distance'' vector $\bo$, which records a distance $o_{t,x} \in R_t$ for every vertex $x \in X_b$ and every metric $t \in [T]$. (This distance represents an assumption that vertex $x \in X_b$ is within distance $o_{t,x}$ of some  open center outside of the subtree $\cT_B$. )
    \item An ``inner distance'' vector $\bi$, which records a distance $i_{t,x} \in R_t$ for every vertex $x \in X_b$ and every metric $t \in [T]$.
    (This distance represents a requirement that we open some center in $\cT_b$ within distance $i_{t,x}$ of vertex $x \in X_b$.)
\end{itemize}
The DP entry $\DP[b, k', \bo, \bi]$ represents the minimum cost paid by clients in $\cT_b$ if we open $k'$ centers in $\cT_b$, while making the assumption that we have already opened facilities outside $\cT_b$ as specified by the interface $\bo$ and while obeying the restriction that we must open centers in $\cT_b$ to satisfy $\bi$. Formally,
\begin{equation}
\label{eq:dp-def}
  \DP[b, k', \bo, \bi] \coloneqq \min_{\substack{S \subseteq F \cap V_b: |S| = k' \\ \forall t \in [T], \forall x \in X, \dist_t(x,S) = i_{t,x}}} 
  \cost(V_b, S, \bo)   
\end{equation}
where
\begin{equation}
    \cost(b, S, \bo)  \coloneqq \sum_{t \in [T]} \sum_{c \in V_b} w_t(c) \cdot \min \left(\dist_t(c, S), \min_{x \in X_b} ({\dist_t(c, x) + o_{t,x}})\right).
\end{equation}
If no such $S$ exists then we set $\DP[b, k', \bo, \bi] = \infty$.


\medskip \noindent \textbf{Objective.} 
Observe that if $b$ is the root bag of the tree decomposition, then the cost of the optimal aggregate solution $\OPT$ is precisely $\min_{\bi} \DP[b, k, \bangle{\infty, \ldots, \infty}, \bi]$.




\medskip \noindent \textbf{Base Cases.} Our base case is when $b$ is a leaf, and in this case we simply evaluate the DP definition by enumeration.  We guess the $k'$ centers $S \subseteq X_b$ that are open; if there is no guess that satisfies the $\bi$ constraint then return $\infty$, otherwise we find the $S$ that minimizes the cost paid by clients in $X_b$. This can be done in $2^{O(\tw)}\cdot\poly(n, T)$ time.



\medskip \noindent \textbf{Recurrence.} Let the two children of $b$ be $\ell$ and $r$. We recurse on $\ell$ and $r$ with distance vectors that are \emph{consistent} per the following definition.

\begin{definition}
\label{def:consistent}
Let $b$ node of $\cT$ with children $l$ and $r$.
Let $(k', \bo, \bi)$ be a tuple that defines a DP entry associated with $b$, consisting of an inner vector of distances, an outer vector of distances, and a budget $k'$. Let $({k'}^\ell, \bo^\ell, \bi^\ell)$ and $({k'}^r, \bo^r, \bi^r)$ be tuples defining DP entries associated with $\ell$ and $r$ respectively.
We say that $({k'}^\ell, \bo^\ell, \bi^\ell, {k'}^r, \bo^r, \bi^r)$ is consistent with $(\bo,\bi,\bc,k')$ if the following conditions all hold.
\begin{itemize}
    \item $k'_l + k'_r = k'$.
    \item \textnormal{[Constraint $\bi$ is satisfied.]} For every metric $t \in [T]$ and vertex $x \in X_b$,
    there exists some $y \in X_\ell$ (or some $y \in X_r$) such that $i_{t,x} = \dist_t(x, y) + i^\ell_{t, t}$ (or, respectively, $i_{t,x} = \dist_t(x, t) + i^r_{t, t}$). In other words, the restriction $\bi$ for bag $b$ is satisfied by a restriction that either $\bi^\ell$ or $\bi^r$ imposes on a child bag.
    %
%
    \item \textnormal{[Assumptions $\bo^\ell, \bo^r$ are justified.]} For each $t \in [T]$ and $y \in X_\ell$ (and similarly for $y \in X_r$), 
    there some $x \in X_b$ such that either $o^\ell_{t,y} = \dist_t(y, x) + o_{t, x}$ or $o^\ell_{t, y} = \dist_t(y, x) + i^r_{t,y}$. In other words, the assumption $\bo^r$ on bag $r$ is justified either by the assumption $\bo$ on the parent bag $b$ or by the constraint $\bi$ on the sibling bag $\ell$.
    %

\end{itemize}
\end{definition}

To evaluate $\DP[b, k', \bo, \bi]$, we guess over all $n^{O(T \cdot \tw)}$ possibilities for $\bo^\ell,\bo^r,\bi^\ell,\bi^r,{k'}^\ell,{k'}^{r}$ and verify consistency in $\poly(n)$ time per guess; we return the minimum sum of consistent child DP entries.  In notation:

\begin{equation}
\label{eq:dp-rec}
\DP[b,k'\bo,\bi] \gets \min_{\substack{({k'}^\ell,\bo^\ell,\bi^\ell,{k'}^r, \bo^r,\bi^r) \text{ consistent}\\ \text{with } (k',\bo,\bi)}} \DP[\ell,{k'}^\ell,\bo^\ell,\bi^\ell]
+ \DP[r,{k'}^r,\bo^r,\bi^r].
\end{equation}

Since there are $n^{O(T\cdot\tw)}$ entries, filling the DP table bottom-up leads to a total running time of $n^{O(T\cdot\tw)}$.

\medskip \noindent \textbf{Correctness of recurrence.}
We show that recurrence \eqref{eq:dp-rec} correctly outputs the value of the DP cell according to the definition \eqref{eq:dp-def}.

First we show that $\eqref{eq:dp-rec} \le \eqref{eq:dp-def}$, by observing that
for any $\DP[b,k', \bo, \bi]$, there is a choice of consistent $({k'}^\ell,\bo^\ell,\bi^\ell,{k'}^r, \bo^r,\bi^r)$ such that $\DP[\ell,{k'}^\ell,\bo^\ell,\bi^\ell]
+ \DP[r,{k'}^r,\bo^r,\bi^r] = \DP[b,k', \bo, \bi].$
Indeed, let $S$ be the set of centers minimizing \eqref{eq:dp-def}, such that $\cost(b,S,\bo) = \DP[b,k', \bo, \bi]$. Let $\ell$ and $r$ be the children of bag $b$, and let $S_\ell = V_\ell \cap S$ and $S_r = V_r \cap S$.
We define variables as follows:
\begin{itemize}
    \item Set $k'^\ell \coloneqq |S_\ell|$ and $k'^r \coloneqq |S_r|$.
    \item For every $t \in [T]$ and $y \in X_{\ell}$ (and similarly for $y \in X_{r}$), set $i^\ell_{t, y} \coloneqq \dist_t(y, S_\ell)$,
    \item and set $o^\ell_{t, y} \coloneqq 
\min\paren{d_t(y,S_r), \min_{x \in X_b}\paren{d_t(y,x) + o_x}}$.
\end{itemize}
As $S = S_\ell \cup S_r$, it is immediate that $\cost(b, S, \bo) = \cost(\ell, S_\ell, \bo^\ell) + \cost(r, S_r, \bo^r)$, meaning that $\DP[\ell,{k'}^\ell,\bo^\ell,\bi^\ell]
+ \DP[r,{k'}^r,\bo^r,\bi^r] = \DP[b,k', \bo, \bi].$
It remains to check the consistency of the variables. Consistency follows from properties of tree decomposition: the [Constraint $\bi$ is satisfied] property follows from the fact that any shortest path from $x \in X_b$ to $S_\ell$ (resp. $S_r$) passes through some vertex of $X_\ell$ (resp. $S_r$), and the [Assumptions $\bo^\ell, \bo^r$ are justified] property follows from the fact that any shortest path between $V_\ell$ to $V_r$ passes through some vertex of $X_b$.

We now need to show the other direction, $\eqref{eq:dp-def} \le \eqref{eq:dp-rec}$: that is, every choice of a consistent tuple $({k'}^\ell,\bo^\ell,\bi^\ell,{k'}^r, \bo^r,\bi^r)$ corresponds to some candidate set of centers $S \subseteq V_b$ such that $\cost(b, S, \bo) = \DP[\ell,{k'}^\ell,\bo^\ell,\bi^\ell]
+ \DP[r,{k'}^r,\bo^r,\bi^r]$. Indeed, we can choose $S \coloneqq S_\ell \cup S_r$, where $S_\ell$  (resp. $S_r$) is the set of centers minimizing $\DP[\ell,{k'}^\ell,\bo^\ell,\bi^\ell]$ (resp. $\DP[r,{k'}^r,\bo^r,\bi^r]$). The set $S$ is a valid candidate set of centers: the size of $S$ is $k'^\ell + k'^r = k'$, and (by properties of tree decomposition and the [Constraint $\bi$ is satisfied] condition of consistency) for every $t \in [T]$ and $x \in X_b$  we have $\dist_t(x, S) = i_{t, x}$. Moreover, we have that $\cost(b, S, \bo) = \cost(\ell, S_\ell, \bo^\ell) + \cost(r, S_r, \bo^r)$; the properties of tree decomposition and the [Assumptions $\bo^\ell, \bo^r$ are justified] property of consistency implies that for every vertex $c \in V_\ell$ (resp. $V_r$), the cost paid by $c$ under $(S, \bo)$ is the same as the cost paid by $c$ under $(S_\ell, \bo^\ell)$ (resp. $(S_r, \bo^r)$).

\medskip \noindent \textbf{$1+\eps$ approximation in FPT time.} We have described an exact algorithm that runs in time $n^{O(T \cdot \tw)}$. We now sketch how to improve the running time, if one allows $1+\eps$ approximation.
By a standard reduction, we can assume that that aspect ratio (ie, the ratio of the largest distance to smallest distance) in each metric is $\poly(n, \eps^{-1})$ at the cost of $1+\eps$ approximation. Note that in the exact algorithm, we duplicated some vertices and connected them by 0-weight edges, to enforce the assumption that the facilities and clients appear only in leaf bags --- now that we want bounded aspect ratio, we instead add $\frac\e{\poly(n)}$-weight edges, again losing only a $1+\e$ factor in the approximation. We then round every distance in each metric to the closest power $(1+\frac{\eps}{\log n})$; now there are $O(\log_{1+\frac{\eps}{\log n}} n) = O(\eps^{-1} \log^2 n)$ distinct distances rather than $O(n^2)$. In our DP, rather than guessing $o_{t,x}$ and $i_{t,x}$ from the set of $O(n^2)$ distances, we guess the distances only to the closest power of $(1+\frac{\eps}{\log n})$. When checking consistency in \Cref{def:consistent}, we modify the conditions [Constraint $\bi$ is satisfied] and [Assumptions $\bo^\ell,\bo^r$ are justified] to only require that the equalities hold up to a $1+\frac{\e}{\log n}$ approximation factor.
The overall running time is $(\frac{\log n}{\eps})^{O(T \cdot \tw)} \poly(n)$, which is bounded above by $\eps^{-O(T^2 \cdot \tw^2)} \poly(n)$ by \Cref{fact:loglog-trick}.

Each recursive call of the DP accumulates $1+\frac{\eps}{\log n}$ error due to the rounding. To be more precise, if the recursive evaluation of the DP cell at bag $b$ has value $\kappa$, and the bag $b$ has height $h$ in the tree $\cT$, then one can show (by an inductive argument) that there is a valid set of centers $S \subseteq V_b$ with cost at most $(1+\frac{\eps}{\log n})^h \cdot \kappa$ according to the definition \eqref{eq:dp-def}. As the height of the tree decomposition $\cT$ is $O(\log n)$, overall we get a solution with approximation ratio $1+O(\eps)$. Rescaling $\eps$ by a constant factor yields the claim.



\medskip \noindent \textbf{Generalizing to $(k,z)$-clustering and $k$-supplier.} We can generalize the DP to work with stochastic aggregate $(k,z)$-clustering for any $z$. In the definition of the DP, we now use 
\[\cost(b, S, \bo) \coloneqq \sum_{t \in [T]} \left(\sum_{c \in V_b} w_t(c) \cdot \min(\dist_t(c,S), \min_{x\in X_b}(\dist_t(c,x) + o_{t,x}))^z \right)^{1/z}\]
so that the DP entry at a bag $b$ gives us the $(k,z)$-clustering cost of the subtree rather than $k$-median cost. We modify the recurrence: rather than just taking a sum of the two children DP entries, we set
\[\DP[b,k'\bo,\bi] \gets \min_{\substack{({k'}^\ell,\bo^\ell,\bi^\ell,{k'}^r, \bo^r,\bi^r) \\ \text{consistent with } \\ (k',\bo,\bi)}} \left(\left(\DP[\ell,{k'}^\ell,\bo^\ell,\bi^\ell]\right)^z + \left(\DP[r,{k'}^r,\bo^r,\bi^r]\right)^z\right)^{1/z}.\]
We use this recurrence because in $(k,z)$-clustering we have the property that $\cost(b,S, \bo) = \left(\cost(\ell, S_\ell, \bo^\ell)^z + \cost(r, S_r, \bo^r)^z\right)^{1/z}$, rather than the property that $\cost(b,S,\bo) = \cost(\ell,S_\ell,\bo^\ell) + \cost(r, S_r, \bo^r)$ that we used for $k$-median. The proof of correctness of the $(k,z)$-clustering recurrence is similar to that of the $k$-median recurrence, once we replace the linearity of $k$-median cost with the new property just described. Our $1+\eps$ approximate DP also works for $(k,z)$-clustering, because distorting distances by a factor $1+\eps$ only changes $\cost(b,S,\bo)$ by at most a factor $1+\eps$.

By setting $z = \infty$, we recover the aggregate $k$-supplier problem. To be precise, our definition of the DP becomes
\[\cost(b, S, \bo) \coloneqq \sum_{t \in [T]} \max_{c \in V_b}\left(w_t(c) \cdot \min(\dist_t(c,S), \min_{x\in X_b}(\dist_t(c,x) + o_{t,x})) \right)\]
and 
\[\DP[b,k'\bo,\bi] \gets \min_{\substack{({k'}^\ell,\bo^\ell,\bi^\ell,{k'}^r, \bo^r,\bi^r) \text{ consistent}\\ \text{with } (k',\bo,\bi)}} \max \left(\DP[\ell,{k'}^\ell,\bo^\ell,\bi^\ell], \DP[r,{k'}^r,\bo^r,\bi^r]\right).\]
Correctness follows from the fact that $\cost(b, S, \bo) = \max(\cost(\ell, S_\ell, \bo^\ell), \cost(r, S_r, \bo^r))$.



%% file: main.bbl
\newcommand{\etalchar}[1]{$^{#1}$}
\begin{thebibliography}{CAGK{\etalchar{+}}19}

\bibitem[ABB{\etalchar{+}}23]{ABB+23}
Fateme Abbasi, Sandip Banerjee, Jaros{\l}aw Byrka, Parinya Chalermsook, Ameet Gadekar, Kamyar Khodamoradi, D{\'a}niel Marx, Roohani Sharma, and Joachim Spoerhase.
\newblock Parameterized approximation schemes for clustering with general norm objectives.
\newblock In {\em 2023 IEEE 64th Annual Symposium on Foundations of Computer Science (FOCS)}, pages 1377--1399. IEEE, 2023.

\bibitem[ABM{\etalchar{+}}19]{adamczyk2019constant}
Marek Adamczyk, Jaros{\l}aw Byrka, Jan Marcinkowski, Syed~M Meesum, and Micha{\l} W{\l}odarczyk.
\newblock Constant-factor {FPT} approximation for capacitated $k$-median.
\newblock In {\em 27th Annual European Symposium on Algorithms (ESA 2019)}, pages 1--1. Schloss Dagstuhl--Leibniz-Zentrum f{\"u}r Informatik, 2019.

\bibitem[ADF{\etalchar{+}}16]{abraham2016highway}
Ittai Abraham, Daniel Delling, Amos Fiat, Andrew~V Goldberg, and Renato~F Werneck.
\newblock Highway dimension and provably efficient shortest path algorithms.
\newblock {\em Journal of the ACM (JACM)}, 63(5):1--26, 2016.

\bibitem[AFGW10]{abraham2010highway}
Ittai Abraham, Amos Fiat, Andrew~V Goldberg, and Renato~F Werneck.
\newblock Highway dimension, shortest paths, and provably efficient algorithms.
\newblock In {\em Proceedings of the twenty-first annual ACM-SIAM symposium on Discrete Algorithms}, pages 782--793. SIAM, 2010.

\bibitem[AGGN08]{AnthoGGN2008}
Barbara~M Anthony, Vineet Goyal, Anupam Gupta, and Viswanath Nagarajan.
\newblock A plant location guide for the unsure.
\newblock In {\em {the Symposium on Discrete Algorithms (SODA)}}, 2008.

\bibitem[AGGN10]{AnthoGGN2010}
Barbara Anthony, Vineet Goyal, Anupam Gupta, and Viswanath Nagarajan.
\newblock A plant location guide for the unsure: Approximation algorithms for min-max location problems.
\newblock {\em {Mathematics of Operations Research}}, 35(1):79--101, 2010.
\newblock Preliminary version appeared in SODA 2008~\cite{AnthoGGN2008}.

\bibitem[AGK{\etalchar{+}}04]{AryaGKMMP2004}
Vijay Arya, Naveen Garg, Rohit Khandekar, Adam Meyerson, Kamesh Munagala, and Vinayaka Pandit.
\newblock Local search heuristics for $k$-median and facility location problems.
\newblock {\em {SIAM Journal on Computing (SICOMP)}}, 33(3):544--562, 2004.

\bibitem[AISX24]{AgrawISX2023}
Akanksha Agrawal, Tanmay Inamdar, Saket Saurabh, and Jie Xue.
\newblock Clustering what matters: Optimal approximation for clustering with outliers.
\newblock {\em J. Artif. Int. Res.}, 78, 2024.

\bibitem[ARR98]{ARR98}
Sanjeev Arora, Prabhakar Raghavan, and Satish Rao.
\newblock Approximation schemes for euclidean k-medians and related problems.
\newblock In {\em Proceedings of the thirtieth annual ACM symposium on Theory of computing}, pages 106--113, 1998.

\bibitem[BCF25]{BhattCF25}
Sayan Bhattacharya, Mart{\'\i}n Costa, and Ermiya Farokhnejad.
\newblock Fully dynamic $k$-median with near-optimal update time and recourse.
\newblock In {\em {the Symposium on Theory of Computing (STOC)}}, pages 1166--1177, 2025.

\bibitem[BCLP23]{BhattCLP2023}
Sayan Bhattacharya, Mart{\'\i}n Costa, Silvio Lattanzi, and Nikos Parotsidis.
\newblock Fully dynamic $k$-clustering in $\tilde{O}(k$ update time.
\newblock {\em {Adv. in Neural Infomation Processing Systems (NeurIPS)}}, 36:18633--18658, 2023.

\bibitem[BCLP25]{BhattCLP25}
Sayan Bhattacharya, Mart{\'\i}n Costa, Silvio Lattanzi, and Nikos Parotsidis.
\newblock Fully dynamic $k$-center clustering made simple.
\newblock In {\em {the International Conference on Machine Learning (ICML)}}, 2025.

\bibitem[BEF{\etalchar{+}}23]{BatenEFHJMW23}
MohammadHossein Bateni, Hossein Esfandiari, Hendrik Fichtenberger, Monika Henzinger, Rajesh Jayaram, Vahab Mirrokni, and Andreas Wiese.
\newblock Optimal fully dynamic k-center clustering for adaptive and oblivious adversaries.
\newblock In {\em {the Symposium on Discrete Algorithms (SODA)}}, pages 2677--2727, 2023.

\bibitem[BGI25]{BGI25}
Sujoy Bhore, Ameet Gadekar, and Tanmay Inamdar.
\newblock Coreset strikes back: Improved parameterized approximation schemes for (constrained) k-median/means.
\newblock {\em arXiv preprint arXiv:2504.06980}, 2025.

\bibitem[BH98]{BH98}
Hans~L Bodlaender and Torben Hagerup.
\newblock Parallel algorithms with optimal speedup for bounded treewidth.
\newblock {\em SIAM Journal on Computing}, 27(6):1725--1746, 1998.

\bibitem[BP25]{BP25}
Romain Bourneuf and Marcin Pilipczuk.
\newblock Bounding $\varepsilon$-scatter dimension via metric sparsity.
\newblock In {\em Proceedings of the 2025 Annual ACM-SIAM Symposium on Discrete Algorithms (SODA)}, pages 3155--3171. SIAM, 2025.

\bibitem[BPR{\etalchar{+}}14]{ByrkaPRST2014}
Jaros{\l}aw Byrka, Thomas Pensyl, Bartosz Rybicki, Aravind Srinivasan, and Khoa Trinh.
\newblock An improved approximation for $k$-median, and positive correlation in budgeted optimization.
\newblock In {\em {the Symposium on Discrete Algorithms (SODA)}}, pages 737--756, 2014.

\bibitem[CAGK{\etalchar{+}}19]{cohen2019tight}
Vincent Cohen-Addad, Anupam Gupta, Amit Kumar, Euiwoong Lee, and Jason Li.
\newblock Tight fpt approximations for k-median and k-means.
\newblock In {\em 46th International Colloquium on Automata, Languages, and Programming (ICALP 2019)}, volume 132, pages 42--1. Schloss Dagstuhl--Leibniz-Zentrum fuer Informatik, 2019.

\bibitem[CAHP{\etalchar{+}}19]{CohenHPSS19}
Vincent Cohen-Addad, Niklas Oskar~D Hjuler, Nikos Parotsidis, David Saulpic, and Chris Schwiegelshohn.
\newblock Fully dynamic consistent facility location.
\newblock {\em {Adv. in Neural Infomation Processing Systems (NeurIPS)}}, 32, 2019.

\bibitem[Car11]{Carat1911}
Constantin Carath{\'e}odory.
\newblock {\"U}ber den variabilit{\"a}tsbereich der fourier’schen konstanten von positiven harmonischen funktionen.
\newblock {\em Rendiconti Del Circolo Matematico di Palermo (1884-1940)}, 32(1):193--217, 1911.

\bibitem[CC16]{chekuri2016polynomial}
Chandra Chekuri and Julia Chuzhoy.
\newblock Polynomial bounds for the grid-minor theorem.
\newblock {\em Journal of the ACM (JACM)}, 63(5):1--65, 2016.

\bibitem[CCS24]{ChakrCS2024}
Deeparnab Chakrabarty, Luc Cote, and Ankita Sarkar.
\newblock Fault-tolerant k-supplier with outliers.
\newblock In {\em {the Symposium on Theoretical Aspects of Computer Science (STACS)}}. Schloss Dagstuhl--Leibniz-Zentrum f{\"u}r Informatik, 2024.

\bibitem[CFK{\etalchar{+}}15]{CyganFKLMPPS2015}
Marek Cygan, Fedor~V Fomin, {\L}ukasz Kowalik, Daniel Lokshtanov, D{\'a}niel Marx, Marcin Pilipczuk, Micha{\l} Pilipczuk, and Saket Saurabh.
\newblock {\em Parameterized algorithms}.
\newblock Springer, 2015.

\bibitem[CFS21]{CFS21}
Vincent {Cohen-Addad}, Andreas~Emil Feldmann, and David Saulpic.
\newblock Near-linear time approximation schemes for clustering in doubling metrics.
\newblock {\em Journal of the ACM (JACM)}, 68(6):1--34, 2021.

\bibitem[CGK20]{ChakrGK2020}
Deeparnab Chakrabarty, Prachi Goyal, and Ravishankar Krishnaswamy.
\newblock {The Non-Uniform $k$-Center Problem}.
\newblock {\em {the ACM Transactions on Algorithms (TALG)}}, 2020.
\newblock Preliminary version in ICALP 2016.

\bibitem[CGS18]{ChanGS18}
TH~Hubert Chan, Arnaud Guerqin, and Mauro Sozio.
\newblock Fully dynamic k-center clustering.
\newblock In {\em {the International World Wide Web Conference}}, pages 579--587, 2018.

\bibitem[CGTS99]{ChariGTS1999}
Moses Charikar, Sudipto Guha, {\'E}va Tardos, and David~B Shmoys.
\newblock A constant-factor approximation algorithm for the k-median problem.
\newblock In {\em {the Symposium on Theory of Computing (STOC)}}, 1999.

\bibitem[CGTS02]{ChariGTS2002}
Moses Charikar, Sudipto Guha, {\'E}va Tardos, and David~B Shmoys.
\newblock A constant-factor approximation algorithm for the k-median problem.
\newblock {\em {Journal of Computer and System Sciences}}, 65(1):129--149, 2002.
\newblock Preliminary version appeared in STOC 1999~\cite{ChariGTS1999}.

\bibitem[Chu15]{chuzhoy2015excluded}
Julia Chuzhoy.
\newblock Excluded grid theorem: Improved and simplified.
\newblock In {\em Proceedings of the forty-seventh annual ACM symposium on Theory of Computing}, pages 645--654, 2015.

\bibitem[Chu16]{chuzhoy2016improved}
Julia Chuzhoy.
\newblock Improved bounds for the excluded grid theorem.
\newblock {\em arXiv preprint arXiv:1602.02629}, 2016.

\bibitem[CKMN01]{ChariGMN2001}
Moses Charikar, Samir Khuller, David~M Mount, and Giri Narasimhan.
\newblock Algorithms for facility location problems with outliers.
\newblock In {\em {the Symposium on Discrete Algorithms (SODA)}}, volume~1, pages 642--651. Citeseer, 2001.

\bibitem[CL12]{ChariL2012}
Moses Charikar and Shi Li.
\newblock A dependent lp-rounding approach for the k-median problem.
\newblock In {\em {International Colloquium on Automata, Languages, and Programming}}, pages 194--205. Springer, 2012.

\bibitem[CMV22]{ChlamMV2022}
Eden Chlamt{\'a}{\v{c}}, Yury Makarychev, and Ali Vakilian.
\newblock Approximating fair clustering with cascaded norm objectives.
\newblock In {\em {the Symposium on Discrete Algorithms (SODA)}}, 2022.

\bibitem[CN19]{ChakrN2019}
Deeparnab Chakrabarty and Maryam Negahbani.
\newblock {Generalized Center Problems with Outliers}.
\newblock {\em {the ACM Transactions on Algorithms (TALG)}}, 2019.
\newblock Prelim. version in ICALP 2018.

\bibitem[CNS22]{ChakrNS2022}
Deeparnab Chakrabarty, Maryam Negahbani, and Ankita Sarkar.
\newblock Approximation algorithms for continuous clustering and facility location problems.
\newblock In {\em {European Symposium on Algorithms}}, 2022.

\bibitem[CS03]{ChudaS2003}
Fabi{\'a}n~A Chudak and David~B Shmoys.
\newblock Improved approximation algorithms for the uncapacitated facility location problem.
\newblock {\em {Journal of the ACM (JACM)}}, 33(1):1--25, 2003.

\bibitem[CS19]{ChakrS19}
Deeparnab Chakrabarty and Chaitanya Swamy.
\newblock Approximation algorithms for minimum norm and ordered optimization problems.
\newblock In {\em {the Symposium on Theory of Computing (STOC)}}, pages 126--137, 2019.

\bibitem[CT21]{chuzhoy2021towards}
Julia Chuzhoy and Zihan Tan.
\newblock Towards tight (er) bounds for the excluded grid theorem.
\newblock {\em Journal of Combinatorial Theory, Series B}, 146:219--265, 2021.

\bibitem[DLR22]{DengLR2022}
Shichuan Deng, Jian Li, and Yuval Rabani.
\newblock Approximation algorithms for clustering with dynamic points.
\newblock {\em {Journal of Computer and System Sciences}}, 130:43--70, 2022.

\bibitem[Fei98]{Feige1998}
Uriel Feige.
\newblock A threshold of ln n for approximating set cover.
\newblock {\em {Journal of the ACM (JACM)}}, 45(4):634--652, 1998.

\bibitem[FL11]{feldman2011unified}
Dan Feldman and Michael Langberg.
\newblock A unified framework for approximating and clustering data.
\newblock In {\em Proceedings of the forty-third annual ACM symposium on Theory of computing}, pages 569--578, 2011.

\bibitem[FLNFS21]{FichtLNS21}
Hendrik Fichtenberger, Silvio Lattanzi, Ashkan Norouzi-Fard, and Ola Svensson.
\newblock Consistent k-clustering for general metrics.
\newblock In {\em {the Symposium on Discrete Algorithms (SODA)}}, pages 2660--2678, 2021.

\bibitem[FS25]{ForstS25}
Sebastian Forster and Antonis Skarlatos.
\newblock Dynamic consistent k-center clustering with optimal recourse.
\newblock In {\em {the Symposium on Discrete Algorithms (SODA)}}, pages 212--254, 2025.

\bibitem[GI25]{GI25}
Ameet Gadekar and Tanmay Inamdar.
\newblock {Dimension-Free Parameterized Approximation Schemes for Hybrid Clustering}.
\newblock In Olaf Beyersdorff, Micha{\l} Pilipczuk, Elaine Pimentel, and Nguy\^{e}n~Kim Thang, editors, {\em 42nd International Symposium on Theoretical Aspects of Computer Science (STACS 2025)}, volume 327 of {\em Leibniz International Proceedings in Informatics (LIPIcs)}, pages 35:1--35:20, Dagstuhl, Germany, 2025. Schloss Dagstuhl -- Leibniz-Zentrum f{\"u}r Informatik.

\bibitem[HK56]{HoffmK1956}
AJ~Hoffman and JB~Kruskal.
\newblock Integral boundary points of convex polyhedra, in “linear inequalities and related systems”(hw kuhn and aw tucker, eds.).
\newblock {\em Annals of Mathematical Studies}, 38, 1956.

\bibitem[HS86]{HochbS1986}
Dorit~S Hochbaum and David~B Shmoys.
\newblock A unified approach to approximation algorithms for bottleneck problems.
\newblock {\em {Journal of the ACM (JACM)}}, 33(3):533--550, 1986.

\bibitem[JV01]{JainV2001}
Kamal Jain and Vijay~V. Vazirani.
\newblock Approximation algorithms for metric facility location and $k$-median problems using the primal-dual schema and {Lagrangean} relaxation.
\newblock {\em {Journal of the ACM (JACM)}}, 48(2):274--296, 2001.

\bibitem[KBI19]{KBI19}
Shrinu Kushagra, Shai {Ben-David}, and Ihab Ilyas.
\newblock Semi-supervised clustering for de-duplication.
\newblock In {\em The 22nd International Conference on Artificial Intelligence and Statistics}, pages 1659--1667. PMLR, 2019.

\bibitem[KLS18]{KrishLS2018}
Ravishankar Krishnaswamy, Shi Li, and Sai Sandeep.
\newblock Constant approximation for k-median and k-means with outliers via iterative rounding.
\newblock In {\em {the Symposium on Theory of Computing (STOC)}}, pages 646--659, 2018.

\bibitem[KR07]{KR07}
Stavros~G Kolliopoulos and Satish Rao.
\newblock A nearly linear-time approximation scheme for the euclidean k-median problem.
\newblock {\em SIAM Journal on Computing}, 37(3):757--782, 2007.

\bibitem[Kus20]{Kus20}
Shrinu Kushagra.
\newblock Three-dimensional matching is np-hard.
\newblock {\em arXiv preprint arXiv:2003.00336}, 2020.

\bibitem[{\L}HG{\etalchar{+}}24]{LackiHGJV24}
Jakub {\L}{\k{a}}cki, Bernhard Haeupler, Christoph Grunau, Rajesh Jayaram, and V{\'a}clav Rozho{\v{n}}.
\newblock Fully dynamic consistent k-center clustering.
\newblock In {\em {the Symposium on Discrete Algorithms (SODA)}}, pages 3463--3484, 2024.

\bibitem[Li13]{Li2013}
Shi Li.
\newblock A 1.488 approximation algorithm for the uncapacitated facility location problem.
\newblock {\em {Information and Computation}}, 222:45--58, 2013.

\bibitem[LS16]{LiS2016}
Shi Li and Ola Svensson.
\newblock Approximating $k$-median via pseudo-approximation.
\newblock {\em {SIAM Journal on Computing (SICOMP)}}, 45(2):530--547, 2016.

\bibitem[LV17]{LattaV17}
Silvio Lattanzi and Sergei Vassilvitskii.
\newblock Consistent k-clustering.
\newblock In {\em {the International Conference on Machine Learning (ICML)}}, pages 1975--1984, 2017.

\bibitem[MMR19]{MakarMR19}
Konstantin Makarychev, Yury Makarychev, and Ilya Razenshteyn.
\newblock Performance of johnson-lindenstrauss transform for k-means and k-medians clustering.
\newblock In {\em {the Symposium on Theory of Computing (STOC)}}, pages 1027--1038, 2019.

\bibitem[MV20]{MahabV2020}
Sepideh Mahabadi and Ali Vakilian.
\newblock {(Individual) Fairness for $k$-Clustering}.
\newblock In {\em {the International Conference on Machine Learning (ICML)}}, pages 7925--7935, 2020.

\bibitem[MV21]{MakarV2021}
Yury Makarychev and Ali Vakilian.
\newblock Approximation algorithms for socially fair clustering.
\newblock In {\em {the Conference on Learning Theory (COLT)}}, 2021.

\bibitem[Ple87]{Ples1987}
J{\'a}n Plesn{\'\i}k.
\newblock A heuristic for the p-center problems in graphs.
\newblock {\em {Discrete Applied Mathematics}}, 17(3):263--268, 1987.

\bibitem[PW01]{pach2001embedding}
J{\'a}nos Pach and Rephael Wenger.
\newblock Embedding planar graphs at fixed vertex locations.
\newblock {\em Graphs and Combinatorics}, 17:717--728, 2001.

\bibitem[RS86]{robertson1986graph}
Neil Robertson and Paul~D Seymour.
\newblock Graph minors. v. excluding a planar graph.
\newblock {\em Journal of Combinatorial Theory, Series B}, 41(1):92--114, 1986.

\bibitem[Sch21]{schaefer2021new}
Marcus Schaefer.
\newblock A new algorithm for embedding plane graphs at fixed vertex locations.
\newblock {\em The Electronic Journal of Combinatorics}, pages P4--55, 2021.

\bibitem[STA97]{ShmoysTA97}
David~B Shmoys, {\'E}va Tardos, and Karen Aardal.
\newblock Approximation algorithms for facility location problems.
\newblock In {\em {the Symposium on Theory of Computing (STOC)}}, pages 265--274, 1997.

\bibitem[Ste13]{Stein1913}
Ernst Steinitz.
\newblock Bedingt konvergente reihen und konvexe systeme.
\newblock {\em Journal für die reine und angewandte Mathematik}, 1913.

\bibitem[Swa16]{Swamy2016}
Chaitanya Swamy.
\newblock Improved approximation algorithms for matroid and knapsack median problems and applications.
\newblock {\em {the ACM Transactions on Algorithms (TALG)}}, 12(4):1--22, 2016.

\bibitem[YC15]{YanC2015}
Li~Yan and Marek Chrobak.
\newblock Lp-rounding algorithms for the fault-tolerant facility placement problem.
\newblock {\em {Journal of Discrete Algorithms}}, 33:93--114, 2015.

\bibitem[Zin03]{Zinke03}
Martin Zinkevich.
\newblock Online convex programming and generalized infinitesimal gradient ascent.
\newblock In {\em {the International Conference on Machine Learning (ICML)}}, pages 928--936, 2003.

\end{thebibliography}
